\documentclass[11pt,onecolumn]{article}

\usepackage{times}
\usepackage{amsmath}
\usepackage{amssymb}
\usepackage{multirow}
\usepackage{xspace}
\usepackage{graphicx}
\usepackage{ifpdf}
\usepackage{url,hyperref}
\usepackage{latexsym}
\usepackage{xspace}
\usepackage{color}
\usepackage{makeidx}
\usepackage{wrapfig}
\usepackage{mathtools}
\usepackage{authblk}

\newtheorem{innercustomgeneric}{\customgenericname}
\providecommand{\customgenericname}{}
\newcommand{\newcustomtheorem}[2]{%
  \newenvironment{#1}[1]
  {%
   \renewcommand\customgenericname{#2}%
   \renewcommand\theinnercustomgeneric{##1}%
   \innercustomgeneric
  }
  {\endinnercustomgeneric}
}

\newcustomtheorem{customdiscussion}{Discussion}

\long\def\remove#1{}

\newtheorem{theorem}{Theorem}[section] 
\newtheorem{lemma}[theorem]{Lemma}
\newtheorem{claim}[theorem]{Claim}

\newtheorem{definition}[theorem]{Definition}
\newenvironment{proof}{{\em Proof:}}{\hfill{\hfill\rule{2mm}{2mm}}}

\definecolor{darkred}{rgb}{1, 0.1, 0.3}

\newcommand{\myparagraph}[1]	{{\vspace*{0.08in}\noindent{\bf #1~}}}
\newcommand {\mm}[1] {\ifmmode{#1}\else{\mbox{\(#1\)}}\fi}
\newcommand{\denselist}{\itemsep 0pt\parsep=1pt\partopsep 0pt}
\newcommand{\eps}{{\varepsilon}}

\newcommand{\anotherG}		{{\tilde{G}}}
\newcommand{\Ghat}			{{\widehat{G}}}

\newcommand{\athreshold}		{\tau}
\newcommand{\X}				{\mathcal{X}}
\newcommand{\aball}			{\mathrm{B}}

\newcommand{\tG}				{{G^*}}

\newcommand{\tE}				{E^*}
\newcommand{\myER}			{{Erd\H{o}s--R\'enyi\xspace}}
\newcommand{\dm}				{{L}}

\newcommand{\mys}				{{\mathrm{s}}}

\newcommand{\myprob}			{\mathbb{P}}

\newcommand{\anotherGt}		{{\tilde{G}_{\athreshold}}}

\newcommand{\ourmodel}		{{ER-perturbed random geometric graph\xspace}}
\newcommand{\myassumption}		{{\sf Assumption-A}}
\newcommand{\myRegular}			{{\sf Regularity-cond}}
\newcommand{\mydense}			{{\sf Density-cond}}
\newcommand{\Ehat}			{{\widehat{E}}}
\newcommand{\mygoodE}		{{good-edge}\xspace}
\newcommand{\mybadE}		{{bad-edge}\xspace}
\newcommand{\myWSP}		{{well-separated clique-partitions family\xspace}}
\newcommand{\cliqueP}		{{clique-partition\xspace}}
\newcommand{\Bconst}		{{Besicovitch constant\xspace}}
\newcommand{\myBC}		{{\beta}}
\newcommand{\myRC}		{{\rho}}
\newcommand{\aC}			{{\mathrm{C}}}
\newcommand{\aK}			{{\mathsf{K}}}
\newcommand{\myE}			{{\mathbb{E}}}
\newcommand{\aF}		{{\mathrm{F}}}
\newcommand{\aX}		{{\mathrm{X}}}
\newcommand{\aY}		{{\mathrm{Y}}}

\newcommand{\aH}		{{\mathrm{H}}}

\DeclarePairedDelimiter\ceil{\lceil}{\rceil}
\DeclarePairedDelimiter\floor{\lfloor}{\rfloor}

\DeclareMathOperator*{\argmax}{arg\,max}

\begin{document}

\title{Local cliques in ER-perturbed random geometric graphs}

\author[1]{Matthew Kahle}
\author[2]{Minghao Tian}
\author[2]{Yusu Wang}
\affil[1]{Department of Mathematics, The Ohio State University, USA \\
\texttt{mkahle@math.osu.edu}}
\affil[2]{Computer Science and Engineering Dept., The Ohio State University, USA \\
\texttt{tian.394@osu.edu, yusu@cse.ohio-state.edu}}

\maketitle

\begin{abstract}
Random graphs are mathematical models that have applications in a wide range of domains. Among them, the Erd\H{o}s--R\'enyi  and the random geometric graphs are two models whose theoretical properties (e.g., clique number and the largest component) are perhaps most well studied. 
We are interested in mixed models coming from the overlay of two graph structures. In particular, we study the following model where one adds Erd\H{o}s--R\'enyi\xspace (ER) type perturbation to a random geometric graph. More precisely, assume 
$G_\mathcal{X}^{*}$ is a random geometric graph sampled from a nice measure on a metric space $\mathcal{X} = (X,d)$. The input observed graph $\widehat{G}(p,q)$ is generated by removing each existing edge from  $G_\mathcal{X}^*$ with probability $p$, while inserting each non-existent edge to $G_\mathcal{X}^{*}$ with probability $q$.  We refer to such random $p$-deletion and $q$-insertion as ER-perturbation. Although these graphs are related to the objects in the continuum percolation theory, our understanding of them is still rather limited.  

In this paper we consider a localized version of the classical notion of clique number for the aforementioned \emph{ER-perturbed random geometric graphs}: Specifically, we study the \emph{edge clique number} for each edge in a graph, defined as the size of the largest clique(s) in the graph containing that edge. The clique number of the graph is simply the largest edge clique number. 
Interestingly, given a ER-perturbed random geometric graph, we show that the edge clique number presents two fundamentally different types of behaviors, depending on which ``type'' of randomness it is generated from. 

As an application of the above results, we show that by using a filtering process based on the edge clique number, we can recover the shortest-path metric of the random geometric graph $G_\mathcal{X}^*$ within a multiplicative factor of $3$, from an ER-perturbed observed graph $\widehat{G}(p,q)$, for a significantly wider range of insertion probability $q$ than in previous work. 
\end{abstract}

\setcounter{page}{0}
\newpage

\section{Introduction}
Random graphs are mathematical models which have applications in a wide spectrum of domains. \myER{} graph $G(n,p)$ is one of the oldest and most-studied models for networks \cite{newmanrandom}, constructed by adding edges between all pairs of $n$ vertices with probability $p$ independently. 

It is also called the Poisson random graph since in the limit as $n$ tends to infinity, $G(n,p)$ has a Poisson degree distribution \cite{newman2010networks}. Many global properties of this model are well-studied by using probabilistic method \cite{alon2016probabilistic}, such as the clique number and the phase transition behaviors of connected components w.r.t. parameter $p$. 

Another classical type of random graphs is the random geometric graph $G(\mathbb{X}_n; r)$ introduced by Edgar Gilbert in 1961 \cite{gilbert1961random}. This model starts with a set of $n$ points $\mathbb{X}_n$ randomly sampled over a metric space (typically a cube in $\mathbb{R}^d$) from some probability distribution, and edges are added between all pairs of points within distance $r$ to each other. 
The \myER{} random graphs and random geometric graphs exhibit similar behavior for the Poisson degree distribution; however, other properties, such as the clique number and phase transition (w.r.to $p$ or to $r$), could be very different 
\cite{gupta1999critical,penrose1997,penrose2003random,mcdiarmid2011chromatic}. 

This model has many applications in real world where the physical locations of objects involved play an important role \cite{dettmann2016random}, for example wireless or transportation networks \cite{nekovee2007worm,blanchard2008mathematical}, and protein-protein interaction (PPI) networks \cite{higham2008fitting}.

We are interested in mixed models that ``combine'' both types of randomness together. 
One way to achieve this is to add \myER{} type perturbation (percolation) to random geometric graphs. 
A natural question arises: what are the properties of this type of random graphs?
Although these graphs are related to the continuum percolation theory \cite{meester1996continuum}, our understanding about them so far is still limited: In previous studies, the underlying spaces are typically plane (called the Gilbert disc model) \cite{bollobas2006percolation}, cubes \cite{Coppersmith2002} and tori \cite{janson2016bootstrap};  the vertices are often chosen as the standard lattices of the space; and the results usually concern the connectivity \cite{booth2003ad,peters2015theoretical} or diameter (e.g, \cite{Wu2017Mixing}). 

Cliques in graphs are important objects in many application domains (e.g, in social networks \cite{hanneman2005introduction}, chemistry \cite{national1995mathematical} and PPI networks \cite{srihari2017computational}).  
The clique number in \myER{} random graph has been studied extensively in 20th century \cite{grimmett1975colouring,bollobas1976cliques}. The clique number in random geometric graph is a relative new topic \cite{penrose2003random}. It has dramatically different behaviors when different ranges of $r$ are chosen \cite{penrose2003random,mcdiarmid2011chromatic}. The high dimensional case (when the space where points are sampled from is high dimensional) is also studied \cite{devroye2011high}. More recently, there has been work on the clique number in inhomogeneous random graphs (e.g. power-law graph) \cite{janson2010large,dolevzal2017cliques}.

\myparagraph{Our work.} 
In this paper, we consider a mixed model of \myER{} random graphs and random geometric graphs, and study the behavior of a local property called \emph{edge clique number}. 
More precisely, we use the following \emph{\ourmodel{} model} previously introduced in \cite{ParthasarathyST17}. 
Suppose there is a compact metric space $\X = (X, d)$ (as feature space) with a probability distribution induced by a ``nice'' measure $\mu$ supported on $\X$ (e.g, the uniform measure supported on an embedded smooth low-dimensional Riemannian manifold). 
Assume we now randomly sample $n$ points $V$ i.i.d from this measure $\mu$, and build the random geometric graph $G^{*}_\X(r)$, which is the $r$-neighborhood graph spanned by $V$ (i.e, two points $u,v \in V$ are connected if their distance $d(u, v) \leq r$). 
Next, we add \myER{} (ER) type perturbation to $G^{*}_\X(r)$: each edge in $G^{*}_\X(r)$ is deleted with a uniform probability $p$, while each ``short-cut'' edge between two unconnected nodes $u, v$ is inserted to $G^{*}_\X(r)$ with a uniform probability $q$. We denote the resulting generated graph as $\Ghat_\X^{p,q}(r)$. 

Intuitively, one can imagine that a graph is generated first as a proximity graph (captured by the random geometric graph) in some feature space ($\X$ in the above setting). The random insertion / deletion of edges then allows for noise or exceptions. For example, in a social network, nodes could be sampled from some feature space of people, and two people could be connected if they are nearby in the feature space. However, there are always some exceptions -- friends could be established by chance even they are very different from each other (``far away''), and two similar (``close'') persons (say, close geographically and in tastes) may not develop friendship. The ER-perturbation introduced above by \cite{ParthasarathyST17} aims to account for such exceptions. 

We introduce a local property called the \emph{edge clique number} of a graph $G$, to provide a more refined view than the global clique number. It is defined for each edge $(u,v)$ in the graph, denoted as $\omega_{u,v}(G)$, as the size of the largest clique containing $uv$ in graph $G$. 
Our main result is that $\omega_{u,v}(G)$ presents two fundamentally different types of behaviors, depending on from which ``type'' of randomness the edge $(u,v)$ is generated from: A ``good'' edge from the random-geometric graph $G^{*}_\X(r)$ has an edge-clique number similar to edges from a certain random-geometric graph; while a ``bad'' edge $(u,v)$ introduced during the random-insertion process has an edge-clique number similar to edges in some random \myER{} graph. See Theorems \ref{thm:insertiononlygoodedge}, \ref{thm:insertiononlybadedge}, \ref{thm:combinedgoodclique}, and \ref{thm:combinedbadclique} for the precise statements. 

On the surface, this may not look surprising: Consider the insertion only case (i.e, the deletion probability $p = 0$ during the ER-perturbation stage). For this simpler case, one could view the final observed graph $\Ghat$ as the \emph{union} of a random geometric graph $\tG$ and an \myER{} random graph $G(n, q)$. However, in general, the edge clique number for an edge in the union $G = G_1 \cup G_2$ of two graphs $G_1$ and $G_2$ could be significantly larger than the clique number in each individual graph $G_i$: Consider for example $G_1$ is a collection of $\sqrt{n}$ disjoint cliques, each of size $\sqrt{n}$, while $G_2$ equals to the complement of $G_1$. The union $G_1 \cup G_2$ is the complete graph and the edge clique number for every edge is $n$. However, the largest clique in $G_1$ or in $G_2$ is $\sqrt{n}$. Our results suggest that due to the randomness in each of the individual graph we are considering, with high probability such a scenario will not happen. To prove our technical results, we apply a novel approach using what we call a well-separated clique-partitions family to help us to decouple the interaction between the two types of hidden random structures (i.e, random geometric graph, and the ER-perturbation).

As an application of our theoretical analysis, in Theorem \ref{thm:maincombined}, we show that by using a filtering process based on our edge clique number, we can recover the shortest-path metric of the random geometric graph $G^{*}_\X(r)$ within a multiplicative factor of $3$, from an ER-perturbed graph $\Ghat_\X^{p,q}(r)$, for a significantly wider range of insertion probability $q$ than what's required in \cite{ParthasarathyST17} \footnote{We note that however, for the metric recovery purpose, the filtering process in \cite{ParthasarathyST17} requires only computing the so-called Jaccard index and is thus significantly simpler than using cliques.}, although we do need a stronger regularity condition on the measure $\mu$. For example, in the case where only insertion-type perturbation is added to the random geometric graph $G^*_{\X}(r)$ (in which case, the observed graph $\Ghat_\X^{0,q}(r)$ can simply be thought of as a union of a random geometric graph and an \myER{} graph), depending on the the density of the graph $G_\X^*(r)$, the contrast could be between requiring that $q \leq \frac{\ln n}{n}$ in previous work versus only requiring that $q = o(1)$ by our method. 
See more discussion at the end of Section \ref{sec:metricrecovery}.

\section{Preliminaries}
\label{sec:definition}

Suppose we are given a compact geodesic metric space $\X = (X, d)$ 
\cite{bridson2011metric}
\footnote{A geodesic metric space is a metric space where any two points in it are connected by a path whose length equals the distance between them. Uniqueness of geodesics is not required. Riemannian manifolds or path-connected compact sets in the Euclidean space are all geodesic metric spaces.}. 
We will consider ``nice'' measures on $\X$. Specifically, 

\begin{definition}[Doubling measure]\label{def:doublingdim}
Given a metric space $\X = (X, d)$, let $B_r(x) \subset X$ denotes the closed metric ball $B_r(x) = \{y\in X \mid d(x,y) \leq r\}$. 
A measure $\mu: X \to \mathbb{R}$ on $\X$ is said to be \emph{doubling} if every metric ball (with positive radius) has finite and positive measure and there is a constant $L=L(\mu)$ s.t. for all $x\in X$ and every $r > 0$, we have 
$\mu(B_{2r}(x)) \le L \cdot \mu(B_r(x))$. 
We call $L$ the \emph{doubling constant} and say $\mu$ is an \emph{$L$-doubling measure}. 
\end{definition}

It is known that any metric space supporting a doubling measure is necessarily a \emph{doubling space}. 
Intuitively, the doubling measure generalizes a nice measure on the Euclidean space, but still behaves nicely in the sense that the growth of the mass within a metric ball is bounded as the radius of the ball increases. 

For our theoretical results later, we in fact need a stronger condition on the input measure, which we will specify later in \myassumption{} at the beginning of Section \ref{sec:edgeclique}. 

\myparagraph{ER-perturbed random geometric graph.} Following \cite{ParthasarathyST17}, we consider the following random graph model: 
Given a compact metric space $\X = (X, d)$ and a $\dm$-doubling probability measure $\mu$ supported on $X$, let $V$ be a set of $n$ points sampled i.i.d. from $\mu$. We build the $r-$neighborhood graph $G^{*}_\X(r) = (V, \tE)$ for some parameter $r > 0$ on $V$; that is, $\tE = \{ (u, v) \mid d(u, v) \le r, u, v \in V \}$. 
We call $G^*_\X(r)$ a random geometric graph generated from $(\X, \mu, r)$. 
Now we add the following two types of random perturbations:
\begin{description}\denselist
\item {\sf $p$-deletion}: For each existing edge $(u,v) \in \tE$, we delete edge $(u,v)$ with probability $p$. 
\item {\sf $q$-insertion}: For each non-existent edge $(u,v) \notin \tE$, we insert edge $(u,v)$ with probability $q$. 
\end{description}
The order of applying the above two types of perturbations doesn't matter since they are applied to two disjoint sets respectively. The final graph $\Ghat_\X^{p,q}(r) = (V, \Ehat)$ is called a \emph{($p,q$)-perturbation of $G^*_\X(r)$}, or simply an \emph{ER-perturbed random geometric graph}. 
The reference $\X$ and parameters $r, p, q$ are sometimes omitted from the notations when their choices are clear. 

We now introduce a local version of the standard clique number: 
\begin{definition}[Edge clique number] \label{def:edgecliquenumber}
Given a graph $G = (V,E)$, for any edge $(u,v) \in E$, its \emph{edge clique number $\omega_{u,v}(G)$} is defined as
\begin{align*}
\omega_{u,v}(G) = \text{ the size of the largest clique in $G$ containing } (u,v).
\end{align*}
\end{definition}

\myparagraph{Setup for the remainder of the paper.} In what follows, we fix the compact geodesic metric space $\X=(X,d)$, the $L$-doubling probability measure $\mu$, and the set of $n$ graph nodes $V$ sampled i.i.d from $\mu$. The input is a ($p,q$)-perturbation $\Ghat = \Ghat_\X^{p,q}(r) = (V, \Ehat)$ of a random geometric graph $\tG = G^*_\X(r)$ spanned by $V$ with radius parameter $r$. For an arbitrary graph $G$, let $V(G)$ and $E(G)$ refer to its vertex set and edge set, respectively, and let $N_{G}(u)$ denote the set of neighbors of $u$ in $G$ (i.e. nodes connected to $u \in V(G)$ by edges in $E(G)$).  

\begin{definition}[Good / \mybadE{}s]\label{def:goodedge}
An edge $(u,v)$ in the perturbed graph $\Ghat$ is a \emph{\mygoodE{}} if $d(u,v) \leq r$. 
An edge $(u,v)$ in the perturbed graph $\Ghat$ is a \emph{\mybadE{}} if for any $x \in N_{G^*}(u)$ and $ y \in N_{G^*}(v)$, we have $d(x,y) > r$.
\end{definition}

In other word, $(u,v)$ is a \mybadE{} if and only if there are no edges between neighbors of $u$ and neighbors of $v$ in $\tG$. See figure \ref{fig:edgetypes} for some examples. It is easy to see that any edge $(u,v)$ in $\Ghat$ with $d(u,v) > 3r$ is necessarily a \mybadE{}. 
\begin{figure}[h]
\centering
\begin{tabular}{ccc}
\includegraphics[height=3.0cm]{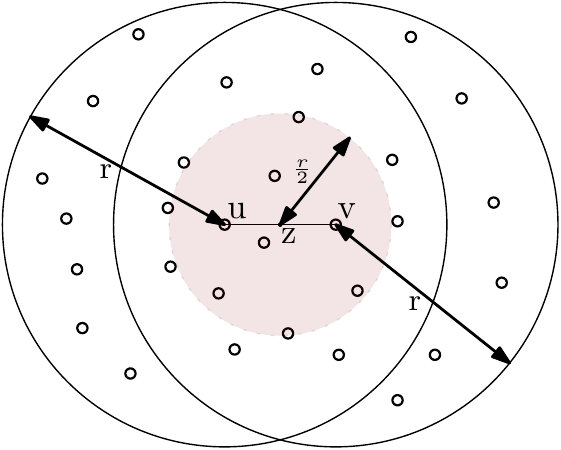} &\hspace*{0.3in} & \includegraphics[height=3.0cm]{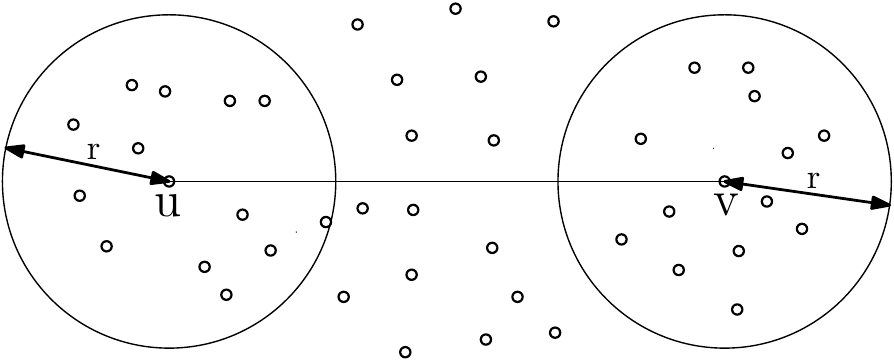} \\
(a) & & (b)
\end{tabular}
\vspace*{-0.15in}
\caption{(a) shows a \mygoodE{} $(u,v)$. It also shows that if $d(u,v) \le r$, then there exists an $r/2$-ball (shaded region) in the intersection of $B_r(u)$ and $B_r(v)$; (b) shows a \mybadE{} $(u,v)$;}
\label{fig:edgetypes}
\end{figure}

\myparagraph{Organization of paper.}  
In Section \ref{sec:edgeclique}, we study the behavior of edge clique number for different types of edges. Our main result Theorem \ref{thm:combinedbadclique} roughly suggests, under certain conditions on the insertion probability $q$,  for a \mygoodE{} $(u,v)$ of $\Ghat_\X^{p,q}(r)$, with high probability, $\omega_{u,v}\left(\Ghat\right)$ has order $\Omega\left(\log_{1/(1-p)}\ln n\right)$; while for a \mybadE{} $(u,v)$, its edge-clique number $\omega_{u,v}\left(\Ghat\right)$ has order $o\left(\log_{1/(1-p)}\ln n\right)$ with high probability. 

To illustrate the main ideas, we will first give results for when only edge-insertion type of perturbations is added to the random geometric graph in Section \ref{subsec:insertion} -- In fact, this case is of independent interest as well. 
An application of our result to recover the shortest-path metric of the hidden geometric graph is given in Section \ref{sec:metricrecovery}. 

\section{Two different behaviors of edge clique number}
\label{sec:edgeclique}

In Section \ref{subsec:insertion}, we study the edge clique numbers for the insertion-only perturbed random geometric graphs, both to illustrate the main ideas, and to show the different behaviors of the edge clique number more clearly. 
We note that this case is of independent interest as well; indeed, the graph generated this way can be thought of as the union of a random geometric graph and an \myER{} graph. To decouple the interaction between them when bounding the clique size, we develop a novel approach using what we call the well-separated clique-partitions family. 
In Section \ref{subsec:deletion}, we study the  case for deletion-only perturbed random geometric graphs, where we only delete each edge independently with probability $p$ to obtain an input graph $\Ghat$.
This case is much simpler, and our main result follows from standard probabilistic methods. Thus we only state the main theorem for the deletion-only case in Section \ref{subsec:deletion}, with proofs in Appendix \ref{appendix:subsec:deletion}. 

Finally, we discuss the combined case of an \ourmodel{} in Section \ref{subsec:combined}.

First, we need an assumption on the parameter $r$ (for the random geometric graph $G^*_\X(r)$), as well as a condition on the measure $\mu$ where graph nodes $V$ are sampled from.

\begin{description}
\item {\sf[\myassumption{}]:} The parameter $r$ and the doubling measure $\mu$ satisfy the following condition:

There exist $\mys \ge \frac{13 \ln n}{n} \left(= \Omega(\frac{\ln n}{n})\right)$ and a constant $\myRC$ such that for any $x\in X$ 
\begin{description}\denselist
\item[(\mydense)] $\mu\left(B_{r/2}(x)\right) \ge \mys. $ 
\item[(\myRegular)] $\mu\left(B_{r/2}(x)\right) \leq \myRC \mys$ 
\end{description}
\end{description}

Intuitively, these two conditions require that for the specific $r$ value we choose, the mass contained inside all radius-$r$ metric balls are similar (within a constant $\myRC$ factor); so the measure $\mu$ is roughly uniform at this scale $r$. 
These conditions can be satisfied when the input measure is the so-called (Ahlfors) $d$-regular measure \cite{heinonen2012lectures}, which is in fact stronger and essentially requires that such a bound on the mass in a metric ball $B_{r'}(x)$ holds for every radius $r'$. 
\mydense{} is equivalent to the Assumption-R in \cite{ParthasarathyST17}. 
It intuitively requires that $r$ is large enough such that with high probability each vertex $v$ in the random geometric graph $G^*_\X(r)$ has degree $\Omega(\ln n)$. Indeed, the following is already known (also see Appendix \ref{appendix:claim:nbhd} for the straightforward proof). 

\begin{claim}[\cite{ParthasarathyST17}]\label{claim:degreebound} 
Under \mydense, with probability at least $1 - n^{-5/3}$, each vertex in $G^*_\X(r)$ has at least $sn/4$ neighbors. 
\end{claim}

\subsection{Insertion-only perturbation} 
\label{subsec:insertion}

Recall $G^*_\X(r) = (V, E)$ is a random geometric graph whose $n$ vertices $V$ sampled i.i.d. from a $L$-doubling probability measure $\mu$ supported on a compact metric space $\X=(X, d)$. 
In this section, we assume that the input graph $\Ghat$ is generated from $\tG = G^{*}_\X(r)$ as follows: First, include all edges of $\tG$ in $\Ghat$. Next, for any $u, v\in V$ with $(u,v) \neq E(\tG)$, we add edge $(u,v)$ to $E(\Ghat)$ with probability $q$. That is, we only insert edges to $\tG$ to obtain $\Ghat$. 

First, for \mygoodE{}s, it is easy to obtain the following result (see Appendix \ref{appendix:thm:insertiononlygoodedge} for the straightforward proof). 

\begin{theorem}\label{thm:insertiononlygoodedge}
Assume \mydense{} holds. 
Let $\tG$ be an $n$-node random geometric graph generated from $(X, d, \mu)$ as described. 
Denote $\Ghat = \Ghat^{q}$ the final graph after inserting each edge not in $\tG$ independently with probability $q$. Then, with high probability, for each \mygoodE{} $(u,v)$ in $\Ghat$, its edge clique number satisfies that $\omega_{u,v}(\Ghat) \geq \mys n/4$.
\end{theorem}

Bounding the edge clique number for \mybadE{}s is much more challenging, due to the interaction between local edges (from random geometric graph) and long-range edges (from random insertions). 
To handle this, we will create a specific collection of subgraphs for $\Ghat$ in an appropriate manner, and bound the edge clique number of a \mybadE{} in each such subgraph. The property of this specific collection of subgraphs is that the union of these individual cliques provides an upper bound on the edge clique number for this edge in $\Ghat$. 
To construct this collection of subgraphs, we will use the so-called Besicovitch covering lemma which has a lot of applications in measure theory \cite{federer2014geometric}. 

First, we introduce some notations. We use a \emph{packing} to refer to a countable collection $\mathcal{B}$ of \emph{pairwise disjoint} closed balls. Such a collection $\mathcal{B}$ is a \emph{packing w.r.t.  a set $P$} 
if the centers of the balls in $\mathcal{B}$ lie in the set $P \subset X$, and it is a \emph{$\delta$-packing} if all of the balls in $\mathcal{B}$ have radius $\delta$. A set $\{A_1, \ldots, A_\ell\}, A_i\subseteq X$, \emph{covers $P$} if $P \subseteq \bigcup_i A_i$. 

\begin{theorem}[Besicovitch Covering Lemma, doubling space version)\cite{kaenmaki2016local}]\label{thm:Besicovitch}
Let $\X=(X,d)$ be a doubling space. Then, there exists a constant $\myBC = \myBC(\X) \in \mathbb{N}$ such that for any $P \subset X$ and $\delta >0$,  there are $\myBC$ number of $\delta$-packings w.r.t. $P$, denoted by $\{\mathcal{B}_1, \cdots, \mathcal{B}_\myBC\}$, whose union also covers P.
\end{theorem}

We call the constant $\myBC(\X)$ above the \emph{Besicovitch constant}. 
Given a set $P$, we say that $A$ is \emph{partitioned into} $A_1, A_2, \cdots, A_k$, if $A = A_1\cup \cdots \cup A_k$ and $A_i \cap A_j = \emptyset$ for any $i\neq j$. 

\begin{definition}[Well-separated clique-partitions family]\label{def:wellseparated}
Consider the random geometric graph $\tG = G^*_{\X}(r)$. A family $\mathcal{P} = \{P_i\}_{i \in \Lambda}$, where $P_i\subseteq V$ and $\Lambda$ is the index set of $P_i$s, forms a \emph{\myWSP}{} of $\tG$ if: 
\begin{enumerate}
\item $V = \cup_{i \in \Lambda} P_{i}$.
\item $\forall i \in \Lambda$, $P_i$ can be partitioned as $P_i= C_1^{(i)} \sqcup C_2^{(i)}\sqcup \cdots \sqcup C_{m_{i}}^{(i)}$ where
\begin{enumerate}
\item[(2-a)] $\forall j \in [1, m_i]$, there exist $\bar{v}^{(i)}_j \in V$ such that $C_j^{(i)} \subseteq B_{r/2}\left(\bar{v}^{(i)}_j\right) \cap V$. 
\item[(2-b)] For any $j_1, j_2 \in [1, m_{i}]$ with $j_1 \neq j_2$, $d_H\left(C_{j_1}^{(i)}, C_{j_2}^{(i)}\right) >r$, where $d_H$ is the Hausdorff distance between two sets in metric space $(X,d)$.
\end{enumerate}
\end{enumerate} 
We also call $C_1^{(i)} \sqcup C_2^{(i)}\sqcup \cdots \sqcup C_{m_{i}}^{(i)}$ a \emph{\cliqueP{} of $P_i$ (w.r.t. $\tG$)}, and its \emph{size} (cardinality) is $m_i$. 
The \emph{size} of the \myWSP{} $\mathcal{P}$ is its cardinality $|\mathcal{P}| = |\Lambda|$. 
\end{definition}

In the above definition, (2-a) implies that each $C_{j}^{(i)}$ spans a clique in $\tG$; thus we call $C_j^{(i)}$ as a \emph{clique} in $P_i$ and $C_1^{(i)} \sqcup C_2^{(i)}\sqcup \cdots \sqcup C_{m_{i}}^{(i)}$ a \cliqueP{} of $P_i$. (2-b) means that there are no edges in $\tG$ between any two cliques of $P_{i}$; thus later, any edge in $\Ghat$ between such cliques must come from $q$-insertion. See figure \ref{fig:wellseparated}. 

This ``well-separateness" of a clique-partition is stronger than having a $r$-packing. Nevertheless, we can apply Theorem \ref{thm:Besicovitch} multiple times to show that 
we can always find a \myWSP{} $\mathcal{P}$ of small cardinality bounded by a constant depending on the \Bconst{} $\myBC(\X)$. 
\begin{figure}[htbp]
\centering
\includegraphics[height=6cm]{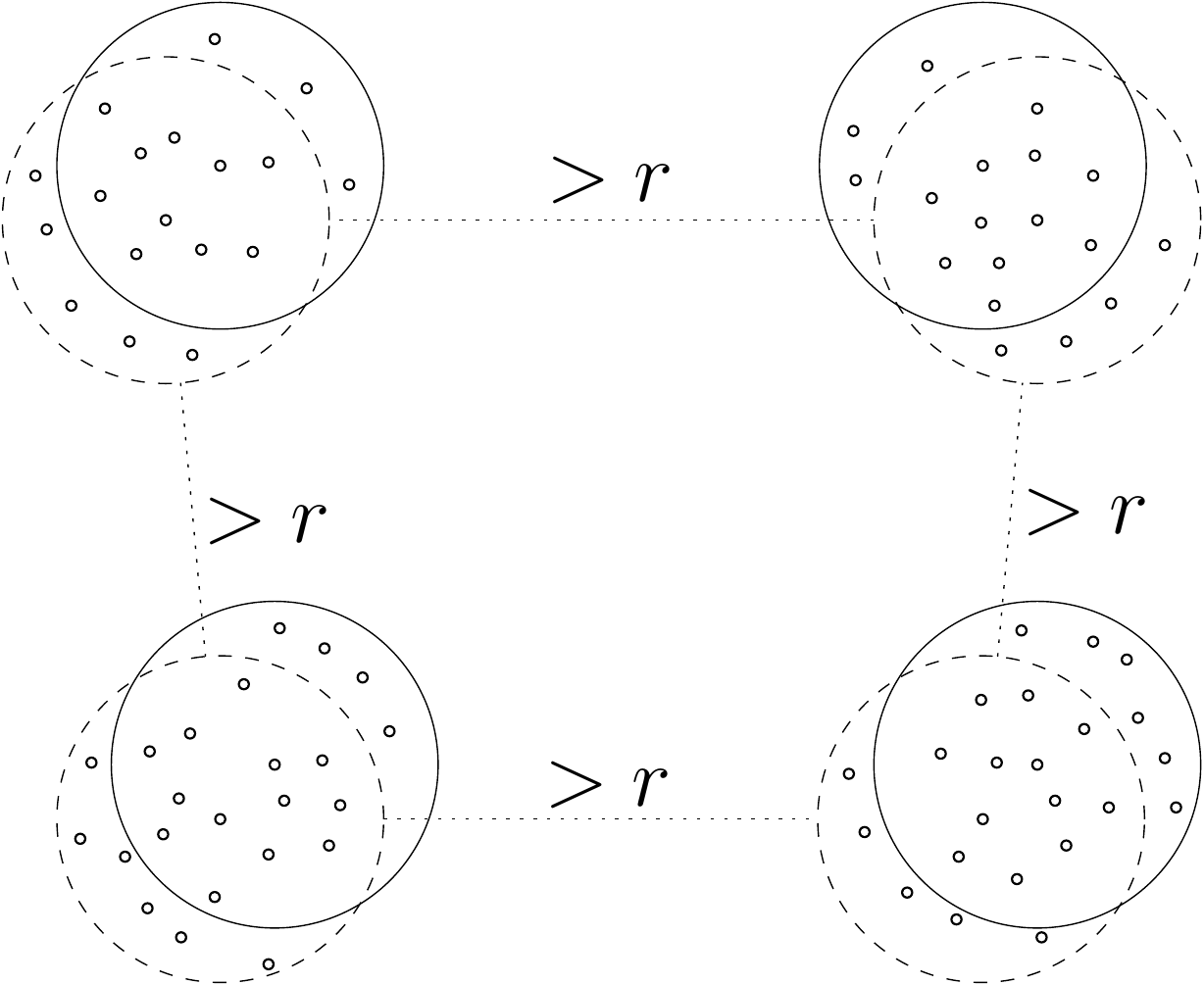}
\caption{Points in the solid balls are $P_1$, and those in dashed balls are $P_2$. Each adapts a \cliqueP{} of size $m_1 = m_2 = 4$. 
Assuming that all nodes in $\tG$ are shown in this figure, then $\mathcal{P} = \{P_1, P_2\}$ forms a \myWSP{} of $\tG$.}
\label{fig:wellseparated}
\end{figure}

\begin{lemma}\label{lem:BCLdoubling}
Let $\tG = G_\X^*(r)$ be an $n$-node random geometric graph generated from $(\X, \mu, r)$ where $\X = (X, d)$ and $\mu$ is a doubling measure supported on $X$. 
There is a \myWSP{} $\mathcal{P}=\{P_i\}_{i \in \Lambda}$ of $\tG$ with $|\Lambda| \leq \myBC^2$, where $\myBC = \myBC(\X)$ is the \Bconst{} of $\X$. 
\end{lemma}
\begin{proof}
To prove the lemma, first imagine we grow an $r/2$-ball around each node in $V \subset X$ (the vertex set of $\tG$). 
By Besicovitch covering lemma (Theorem \ref{thm:Besicovitch}), we have a family of $(r/2)$-packings w.r.t. $V$, $\mathcal{B} = \{\mathcal{B}_1, \mathcal{B}_2, \cdots, \mathcal{B}_{\alpha_1}\}$, whose union covers $V$. Here, the constant $\alpha_1$ satisfies $\alpha_1 \le \myBC(\X)$. 

Each $\mathcal{B}_i$ contains a collection of disjoint $r/2$-balls centered at a subset of nodes in $V$, and let $V_i \subseteq V$ denote the centers of these balls. For any $u, v \in V_i$, we have $d(u, v) > r$ as otherwise, $B_{r/2}(u) \cap B_{r/2}(v) \neq \emptyset$ meaning that the $r/2$-balls in $\mathcal{B}_i$ are not all pairwise disjoint. 
Now consider the collection of $r$-balls centered at all nodes in $V_i$. 
Applying Besicovitch covering lemma to $V_i$ again with $\delta = r$, we now obtain a family of $r$-packings w.r.t. $V_i$, denoted by $\mathcal{D}^{(i)} = \mathcal{D}_1^{(i)} \sqcup \cdots \sqcup \mathcal{D}_{\alpha_2^{(i)}}^{(i)}$, whose union covers $V_i$. Here, the constant $\alpha_2^{(i)}$ satisfies $\alpha_2^{(i)}\le \myBC(\X)$ for each $i \in [1, \alpha_1]$. 

Now each $\mathcal{D}_{j}^{(i)}$ contains a set of disjoint $r$-balls centered at a subset of nodes $V_j^{(i)} \subseteq V_i$ of $V_i$. 
First, we claim that $\bigcup_j V_j^{(i)} = V_i$. This is because that $\mathcal{B}_i$ is an $r/2$-packing which implies that $d(u,v) > r$ for any two nodes $u, v \in V_i$. In other words, the $r$-ball around any node from $V_i$ contains no other nodes in $V_i$. 
As the union of $r$-balls $ \mathcal{D}_1^{(i)} \sqcup \cdots \sqcup \mathcal{D}_{c_2^{(i)}}^{(i)}$ covers $V_i$ by construction, it is then necessary that each node $V_i$ has to appear as the center in at least one $\mathcal{D}_j^{(i)}$ (i.e, in $V_j^{(i)}$). Hence $\bigcup_j V_j^{(i)} = V_i$.

Now for each vertex set $V_j^{(i)}$, let $P_j^{(i)} \subseteq V$ denote all points from $V$ contained in the $r/2$-balls centered at points in $V_j^{(i)}$. As $\cup_j V_j^{(i)} = V_i$, we have that $\bigcup_j P_j^{(i)} = \bigcup_{v\in V_i} \left(B_{r/2}(v) \cap V\right)$. It then follows that $\bigcup_{i \in [1, \alpha_1]}\big( \bigcup_{j\in [1, \alpha_2^{(i)}]} P_j^{(i)} \big) = V$ as the union of the family of $r/2$-packings $\mathcal{B} = \{\mathcal{B}_1, \mathcal{B}_2, \cdots, \mathcal{B}_{c_1}\}$ covers all points in $V$ (recall that $\mathcal{B}_i$ is just the set of $r/2$-balls centered at points in $V_i$). 

Clearly, each $P_j^{(i)}$ adapts a \cliqueP{}: Indeed, for each $V_j^{(i)}$, any two nodes in $V_j^{(i)}$ are at least distance $2r$ apart (as the $r$-balls centered at nodes in $V_j^{i}$ are disjoint), meaning that the $r/2$-balls around them are more than $r$ (Hausdorff-)distance away. 
In other words, $\mathcal{P} = \left\{ P_j^{(i)}, i\in [1, \alpha_1], j \in [1, \alpha_2^{(i)}] \right\}$ forms a \myWSP{} of $\tG$.  
Finally, since $\alpha_1, \alpha_2^{(i)} \le \myBC(\X) = \beta$, the cardinality of $\mathcal{P}$ is thus bounded by $\beta^2$. 
\end{proof}

Using Chernoff bound and the {\sf \myassumption{}}, we can also upper-bound the number of points in every $r/2$-ball centered at nodes of $\tG$. 
The straightforward proof is in Appendix \ref{appendix:claim:r/2ballupperbound}. 

\begin{claim}\label{lm:r/2ballupperbound}
Given an $n$-node random geometric graph $G^* = (V, E^*)$ generated from $(\X, \mu, r)$, if {\sf Assumption-A} holds, then with probability at least $1-n^{-5}$, for every $v\in V$, the ball $B_{r/2}(v) \cap V$ contains at most $3\myRC sn$ points. 
\end{claim}

We now state one of our main theorems, which relates the edge clique number for \mybadE{}s with the insertion probability. 
To simplify notations, we call a clique containing an edge $(u,v)$ a \emph{$uv$-clique}.  

\begin{theorem}\label{thm:insertiononlybadedge}
Let $\tG= G^*_{\X}(r)$ be an $n$-node random geometric graph generated by $(\X, \mu, r)$ where $\mu$ is an $L$-doubling measure supported on $X$. Suppose {\sf \myassumption{}} holds. Let $\Ghat = \Ghat^{q}$ denote the graph obtained by inserting each edge not in $\tG$ independently with probability $q$. Then there exist constants $c_1, c_2, c_3 > 0$ which depend on the doubling constant $L$ of $\mu$, the \Bconst{} $\myBC(\X)$, and the regularity constant $\myRC$, such that for any $\aK = \aK(n)$ with $\aK\rightarrow \infty$ as $n \rightarrow \infty$, 
with high probability, $\omega_{u,v}(\Ghat) <\aK$ for any \mybadE{} $(u,v)$ in $\Ghat$, as long as $q$ satisfies 
\begin{align}\label{eqn:qbound}
q ~~\leq~~ \min\left\{~c_1,~~ c_2 \cdot \left(\frac{1}{n}\right)^{c_3/\aK}\cdot \frac{\aK}{\mys n} \right\}.
\end{align}
Note that this statement holds for any $(u,v)$ with $d(u,v) > 3r$, as such an edge $(u,v)$ must be a \mybadE{}. 
\end{theorem}

\myparagraph{Remark. } 
To illustrate the above theorem, consider for example when $\aK = \Theta(sn)$. Then the theorem says that there exists constant $c'$ such that if $q < c'$, then w.h.p. $\omega_{u,v} < \aK$ (thus $\omega_{u,v} = O(sn)$) for any \mybadE{} $(u,v)$. 
Now consider when $q = o(1)$. Then the theorem implies that w.h.p. the edge-clique number for any \mybadE{} is at most $\aK = o(sn)$ \footnote{To see this, note that the term $\left(\frac{1}{n}\right)^{\frac{c_3}{\aK}}\cdot \frac{\aK}{\mys n}$ is increasing as $\aK$ increases, and combine this with the fact that $\left(\frac{1}{n}\right)^{\frac{c_3}{\aK}}\cdot \frac{\aK}{\mys n} = \Theta(1)$ if $\aK = \Theta(sn)$.}. 
This is qualitatively different from the edge-clique number for a \mygoodE{} for the case $q=o(1)$, which is $\Omega(sn)$ as shown in Theorem \ref{thm:insertiononlygoodedge}. 
By reducing this insertion probability $q$, this gap can be made {\bf larger and larger}.   

\myparagraph{Proof of Theorem \ref{thm:insertiononlybadedge}.} 
Given any node $y$, let $B^V_r(y) \subseteq V$ denote $B_r(y) \cap V$.  
Now consider a \mybadE{} $(u,v)$. Set $A_{uv} = \left\{w \in V | w \notin B_r(u)\cup B_r(v)\right\}$ and $B_{uv} = \{w \in V | w \in B^V_r(u)\cup B^V_r(v)\}$. Denote $\tilde{A}_{uv} = A_{uv} \cup \{u\} \cup \{v\}$; 
easy to check that $V = \tilde{A}_{uv} \cup {B}_{uv}$. 

Let $G|_S$ denote the subgraph of $G$ spanned by a subset $S$ of its vertices. 
Given any set $\aC$, let $\aC|_S  = \aC \cap S$ be the restriction of $\aC$ to another set $S$. 
Now consider a subset of vertices $\aC \subseteq V$: obviously, $\aC = \aC|_{\tilde{A}_{uv}}\cup \aC|_{B_{uv}}$. 

Hence by the pigeonhole principle and the union bound, we have: 
\begin{align}\label{eqn:mainresult}
&\myprob\left[ \Ghat ~\text{has a}~uv\text{-clique of size} \ge \aK \right] \nonumber \\
\le ~~ &\myprob\left[ \Ghat|_{\tilde{A}_{uv}}~\text {has a}~uv\text{-clique of size} \ge \frac{\aK}{2}\right] ~~+~~ \myprob\left[ \Ghat|_{B_{uv}}~\text {has a}~uv\text{-clique of size} \ge \frac{\aK}{2}\right] 
\end{align}
Next, we will bound the two terms on the right hand side of Eqn. (\ref{eqn:mainresult}) separately in Case (A) and Case (B) below. 

\begin{figure}[h]
\centering
\begin{tabular}{ccc}
\includegraphics[height=3.5cm]{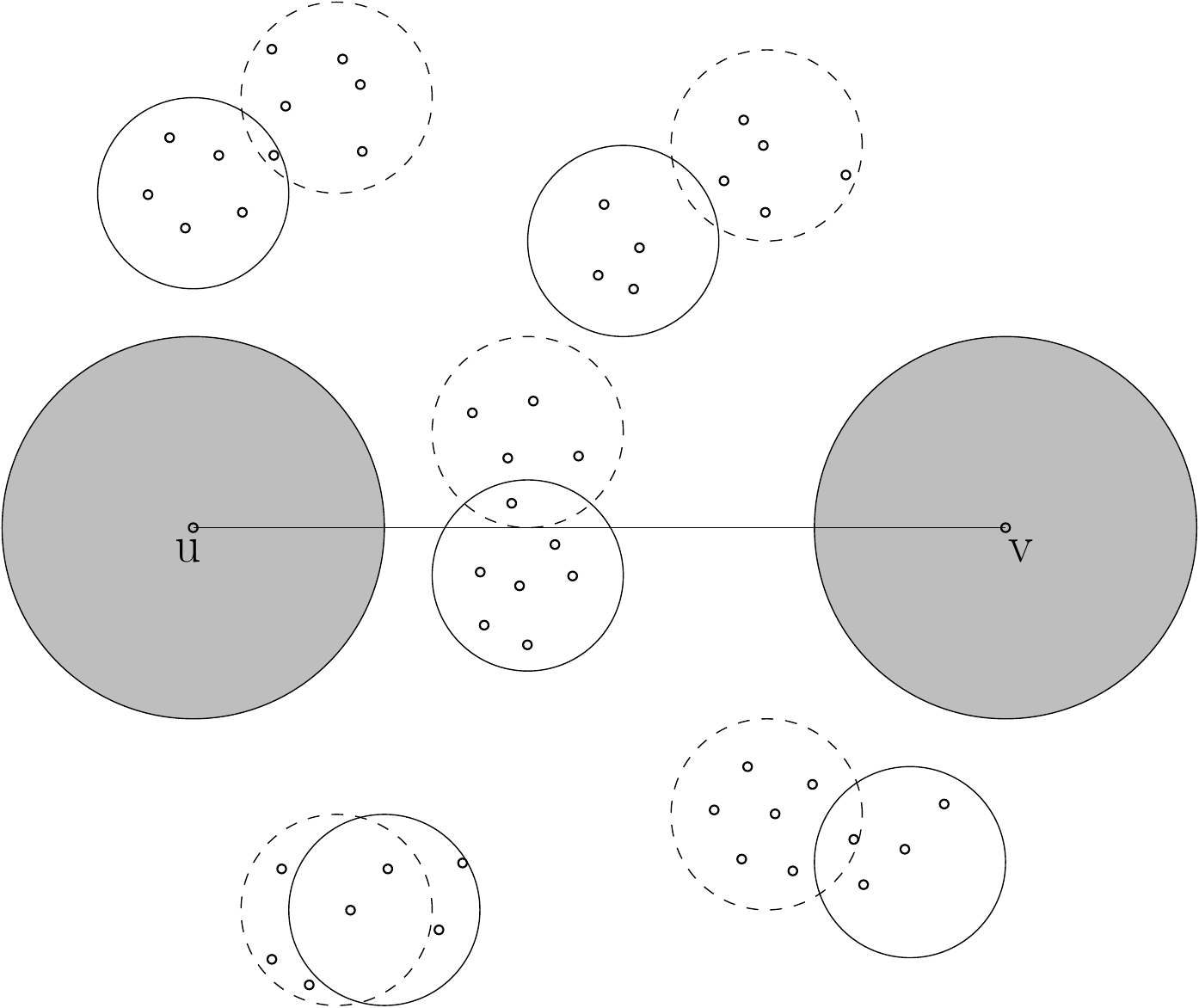} &\hspace*{0.3in}  & \begin{minipage}[c]{6cm}
\vspace*{-3.5cm}
\includegraphics[height=2cm]{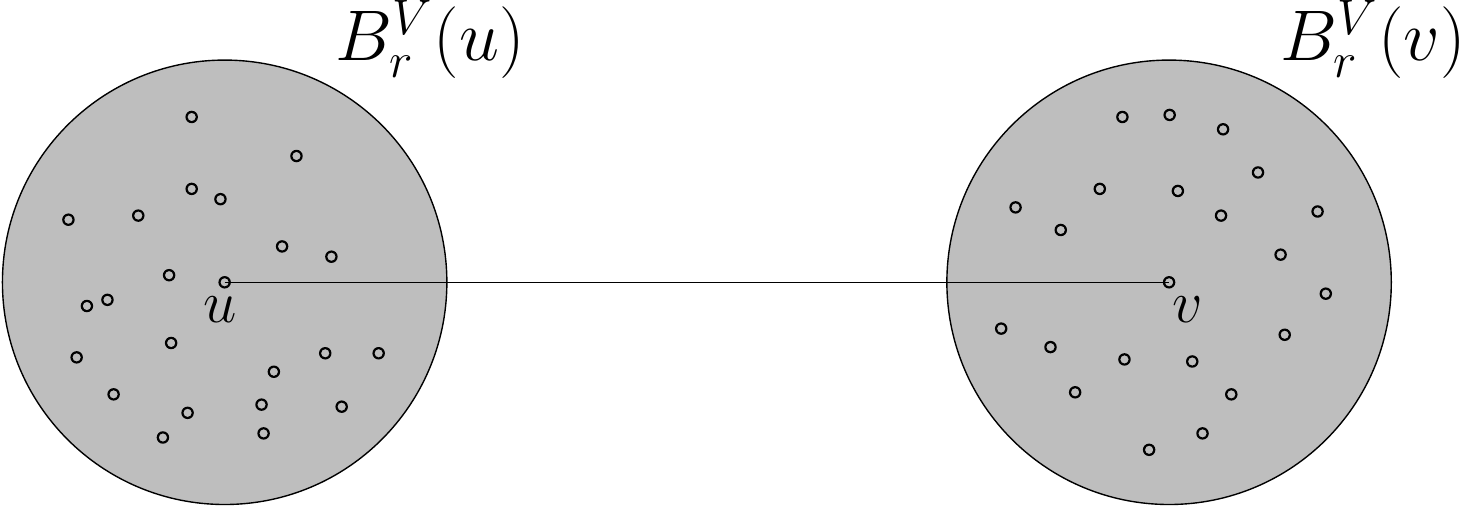}
\end{minipage}  \\
(a) &  & \hspace*{0in}(b)
\end{tabular}
\vspace*{-0.15in}
\caption{(a) A well-separated clique partition $\mathcal{P} = \{P_1, P_2\}$ of $A_{uv}$ --- points in the solid ball are $P_1$, and those in dashed ball are $P_2$. (b) Points in $B_{uv}$. }
\label{fig:twocases}
\end{figure}

\myparagraph{Case (A): bounding the first term in Eqn. (\ref{eqn:mainresult}).} 
We apply Lemma \ref{lem:BCLdoubling} for points in $A_{uv}$.   
This gives us a \myWSP{} $\mathcal{P} = \{P_i\}_{i\in \Lambda}$ of $A_{uv}$ with $|\Lambda| \le \myBC^2$ being a constant. 
See Figure \ref{fig:twocases} (a).
Augment each $P_i$ to $\tilde{P}_i = P_i \cup \{u\} \cup \{v\}$. 
Suppose there is a clique $C$ in $\Ghat|_{\tilde{A}_{uv}}$, then as $\bigcup_i \tilde{P}_i = \tilde{A}_{uv}$, we have $C = \bigcup_{i\in \Lambda} C|_{\tilde{P}_i}$, implying that $|C| \le \sum_{i\in \Lambda} \left| C|_{\tilde{P}_i}\right|$. 
Hence by pigeonhole principle and the union bound, we have: 
\begin{align}\label{eqn:caseAunionbound}
&\myprob\left[ \Ghat|_{\tilde{A}_{uv}}~\text {has a}~uv\text{-clique of size} \ge \frac{\aK}{2}\right] 
~~\leq~~ \sum_{i=1}^{|\Lambda|} ~\myprob\left[\Ghat|_{\tilde{P}_{i}}~\text {has a}~uv\text{-clique of size} \ge \frac{\aK}{2|\Lambda|}\right]
\end{align}    

Now for arbitrary $i \in \Lambda$, consider  $\Ghat|_{\tilde{P}_i}$, the induced subgraph of $\Ghat$ spanned by vertices in $\tilde{P}_i$. Note, $\Ghat|_{\tilde{P}_i}$ can be viewed as generated by inserting each edge not in $\tG|_{\tilde{P}_i} \cup \{uv\}$ to it with probability $q$. 
Recall from Definition \ref{def:wellseparated} that each $P_i$ adapts a clique decomposition $C_1^{(i)}\sqcup \dots \sqcup C_{m_i}^{(i)}$, where every $C_j^{(i)}$ is contained in an $r/2$-ball, and all such balls are $r$-separated (w.r.t Hausdorff distance).

Now fix any $i \in \Lambda$. For simplicity of the argument below, set $m = m_i$, and let $N_j = \left|C_j^{(i)}\right|$ denote the number of points in the $j$-th cluster $C_j^{(i)}$. 
Note that obviously, $m \le |P_i| \le |V|=n$ for any $i\in \Lambda$.
Set $N_{\max} = 3 \rho \mys n$. By Lemma \ref{lm:r/2ballupperbound}, we know that, with high probability (at least $1 - n^{-5}$), $N_j \leq N_{\max}$ for all $j$ in $[1, m]$. 

Denote $\aF$ to be the event that ``for every $v \in V$, the ball $B_{r/2}(v)\cap V$ contains at most $N_{\max}$ points''; and $\aF^c$ denotes the complement event of $\aF$. 
Observe that the induced subgraph $\Ghat|_{\tilde{P}_i}$ consists of a set of cliques (each clique is spanned by some $C_j^{(i)}$), $u$, $v$, edge $uv$, as well as newly inserted edge between them with insertion probability $q$ (see Figure \ref{fig:illustration_case_A}).

\begin{figure}[htbp]
\centering
\includegraphics[height=6cm]{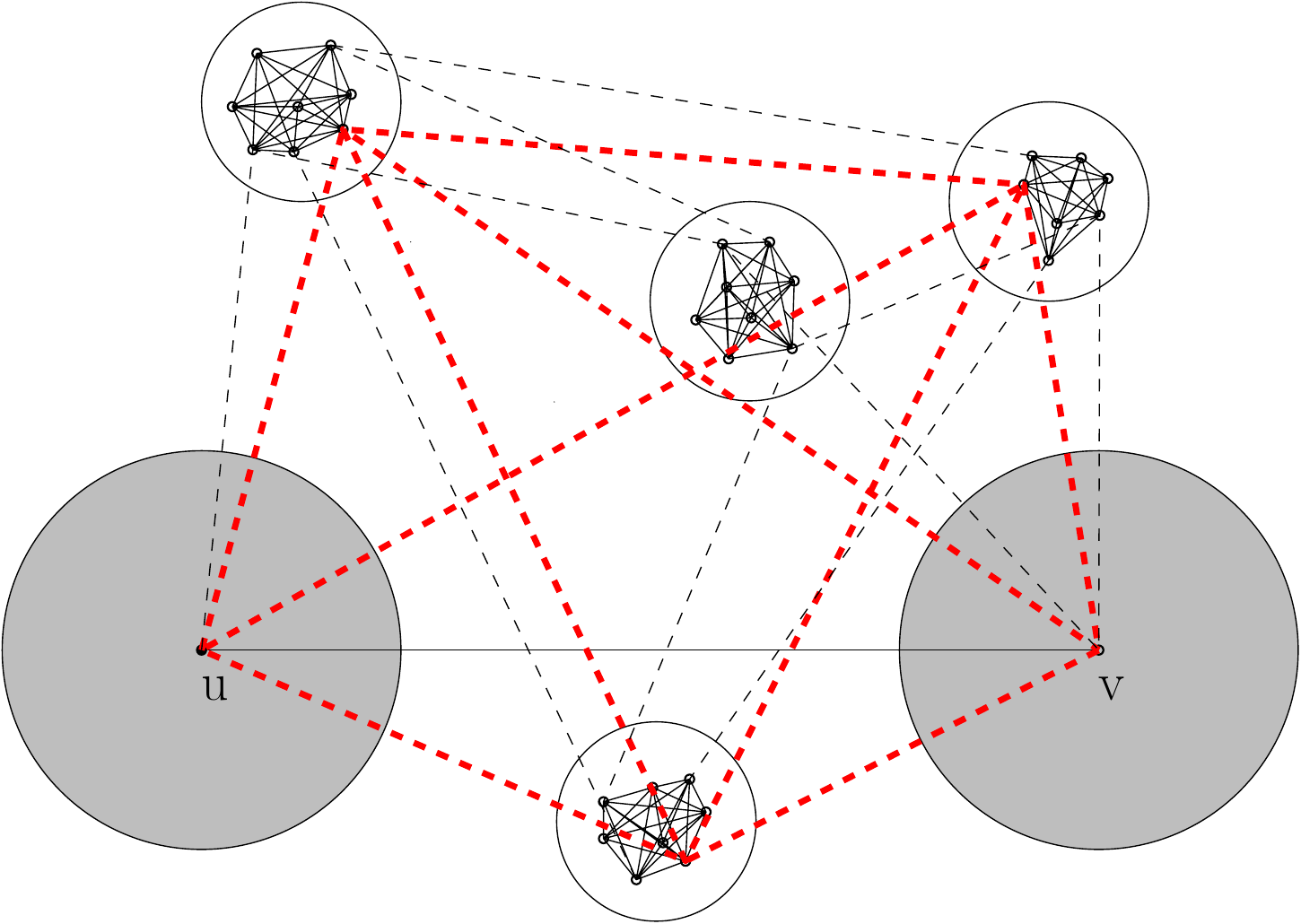}
\caption{The red dashed lines and the edge $uv$ form a possible $K_5$ $uv-$clique in some well-separated clique partition $P_i$. The points in the small balls are the nodes falling in $r/2-$balls. All the dashed lines are the possible inserted edges (all non-existing edges are inserted independently with probability $q$).}
\label{fig:illustration_case_A}
\end{figure}

Now set $k := \floor*{\frac{\aK}{2|\Lambda|}}-2$. For every set $S$ of $k+2$ vertices in this graph $\Ghat|_{\tilde{P}_i}$, let $A_S$ be the event ``\emph{$S$ is a $uv$-clique in $\Ghat|_{\tilde{P}_i}$}'' and $\aX_S$ its indicator random variable. Set 
\begin{align*}
\aX = \sum\limits_{|S|=k+2}\aX_S
\end{align*}
and note that $\aX$ is the number of $uv$-cliques of size $(k+2)$ in $\Ghat|_{\tilde{P}_i}$. 
It follows from Markov inequality that:  
\begin{align}\label{eqn:probGPi}
\myprob\left[\Ghat|_{\tilde{P}_{i}}~\text {has a}~uv\text{-clique of size} \ge  k+2\right] = \myprob[\aX > 0] &=\myprob[\aX>0 \mid \aF] \cdot \myprob[\aF] + \myprob[\aX>0 \mid \aF^c] \cdot \myprob[\aF^c] \nonumber\\
&\leq \myprob[\aX>0 \mid \aF] + \myprob[\aF^c] \leq \myE[\aX \mid \aF] + n^{-5}. 
\end{align}
On the other hand, using linearity of expectation, we have:
\begin{align}\label{expectationX}
\myE[\aX \mid \aF] ~~=~~ \sum\limits_{|S|=k+2}\myE[\aX_S\mid \aF] ~~=~~ q^{2k}\sum\limits_{\substack{x_1 + x_2 + \cdots + x_m = k \\ 0 \leq x_i \leq N_i}} \binom{N_1}{x_1} \binom{N_2}{x_2}\cdots \binom{N_m}{x_m}q^{(k^2-\sum_{i=1}^{m}x_i^2)/2}\nonumber\\
\leq q^{2k}\sum\limits_{\substack{x_1 + x_2 + \cdots + x_m = k \\ 0 \leq x_i \leq N_{\max}}} \binom{N_{\max}}{x_1} \binom{N_{\max}}{x_2}\cdots \binom{N_{\max}}{x_m}q^{(k^2-\sum_{i=1}^{m}x_i^2)/2}
\end{align}
  
To estimate this quantity, we have the following lemma:  
\begin{lemma}\label{lem:insertioncaseA}
There exists a constant $c > 0$ depending on $\myBC$ and $\rho$ such that for any constant $\epsilon>0$, if $\aK \leq c \mys n$ and
\begin{align}\label{eqn:qbound1}
q ~~\leq~~ \min \left\{\left(\dfrac{k!}{n^{\epsilon}N_{\max}^km}\right)^{1/2k},~~ \left(\dfrac{k!}{k^2n^{\epsilon}N_{\max}^km^2}\right)^{1/k},~~ \left(\dfrac{k!}{n^{\epsilon}m^kN_{\max}^{k}}\right)^{4/k^2}\right\}
\end{align}
then we have that $\myE[\aX \mid \aF] = O(n^{-\epsilon})$.

Specifically, we can set $\epsilon = 3$ (this choice will be necessary later to apply union bound) and obtain $\myE[\aX\mid \aF] = O(n^{-3})$. 
\end{lemma}

The proof of this lemma is rather technical, and can be found in Appendix \ref{appendix:lem:insertioncaseA}. 

Furthermore, $|\Lambda| \le \myBC^2$ (which is a constant) and $m = |P_i| \le |V| = n$. 
One can then verify that there exist constants $c^a_2$ and $c^a_3$ (which depend on the doubling constant $L$ of $\mu$, the \Bconst{} $\myBC$, and the regularity constant $\myRC$), such that if 
\begin{align*}
q \le c^a_2 \cdot \left(\frac{1}{n}\right)^{c^a_3/\aK}\cdot \frac{\aK}{\mys n},
\end{align*}
then the conditions in Eqn. (\ref{eqn:qbound1}) will hold (the simple proof of this can be found in Appendix \ref{appendix:constantsc2c3}). 
Thus, by Lemma \ref{lem:insertioncaseA} (with $\epsilon$ set to be $3$) and Eqn. (\ref{eqn:probGPi}), we know that
\begin{align}
\text{If}~~\aK \leq c \mys n ~&\text{and}~~ q \leq c^a_2 \cdot \left(\frac{1}{n}\right)^{c^a_3/\aK}\cdot \frac{\aK}{\mys n}, \nonumber \\
&\text{then}~~ \forall i \in \Lambda, \myprob\left[\Ghat|_{\tilde{P}_{i}}~\text {has a}~uv\text{-clique of size} \ge  k+2\right] = O(n^{-3}). \label{eqn:qone}
\end{align}

On the other hand, note that
\begin{align*}
&\myprob\left[\Ghat|_{\tilde{P}_{i}}~\text {has a}~uv\text{-clique of size} \ge \frac{\aK}{2|\Lambda|}\right] ~~=~~ \myprob\left[\Ghat|_{\tilde{P}_{i}}~\text {has a}~uv\text{-clique of size} \ge  k+2\right]
\end{align*}

As $|\Lambda|$ is a constant, by Eqn. (\ref{eqn:caseAunionbound}), we obtain that
\begin{align}
\text{if}~~\forall i \in &\Lambda,  \myprob\left[\Ghat|_{\tilde{P}_{i}}~\text {has a}~uv\text{-clique of size} \ge \frac{\aK}{2|\Lambda|}\right] = O(n^{-3}) \text{, ~then}~ \nonumber \\
&\myprob\left[\Ghat|_{\tilde{A}_{uv}}~\text {has a}~uv\text{-clique of size} \ge \frac{\aK}{2}\right] = O(n^{-3}) \label{eqn:qtwo}
\end{align}

It then follows from Eqn. (\ref{eqn:qone}) and (\ref{eqn:qtwo}) that 
\begin{align}\label{eqn:caseAalmostfinal}
\text{If}~ \aK \leq c \mys n ~\text{and}~ q \leq c^a_2 \cdot \left(\frac{1}{n}\right)^{c^a_3/\aK}\cdot \frac{\aK}{\mys n}, ~~\text{then}~\myprob\left[ \Ghat|_{\tilde{A}_{uv}}~\text {has a}~uv\text{-clique of size} \ge \frac{\aK}{2}\right] = O(n^{-3}). 
\end{align}
Finally, suppose $\aK > \aK_0 = c sn$. Set 
$$c_1^a ~~=~~ c^a_2 \cdot \left(\frac{1}{n}\right)^{c^a_3/(c \ln n)}\cdot \frac{\aK_0}{\mys n} ~~\le~~ c^a_2 \cdot \left(\frac{1}{n}\right)^{c^a_3/\aK_0}\cdot \frac{\aK_0}{\mys n} ,$$ 
where the inequality holds as by \myassumption{} $sn > \ln n$. 
Plugging $\aK_0 = csn$ to the definition of $c_1^a$, it is then easy to see that $c_1^a$ is a positive constant. 
Using Eqn (\ref{eqn:caseAalmostfinal}), we know that if $q \le c_1^a$ and $\aK > \aK_0 = csn$, then $$\myprob\left[ \Ghat|_{\tilde{A}_{uv}}~\text {has a}~uv\text{-clique of size} \ge \frac{\aK}{2}\right] \le \myprob\left[ \Ghat|_{\tilde{A}_{uv}}~\text {has a}~uv\text{-clique of size} \ge \frac{\aK_0}{2}\right] = O(n^{-3}). $$
Combining this with Eqn. (\ref{eqn:caseAalmostfinal}), we thus obtain that: 
\begin{align}\label{eqn:caseAfinal}
\text{If}~ q \leq \min\left\{ c_1^a, ~c^a_2 \cdot \left(\frac{1}{n}\right)^{c^a_3/\aK}\cdot \frac{\aK}{\mys n} \right\}, ~~\text{then}~\myprob\left[ \Ghat|_{\tilde{A}_{uv}}~\text {has a}~uv\text{-clique of size} \ge \frac{\aK}{2}\right] = O(n^{-3}). 
\end{align}

\noindent{\bf Case (B): bounding the second term in Eqn. (\ref{eqn:mainresult}).} 
Recall that $B_{uv} = \{w\in V \mid w\in B_r^V(u) \cup B^V_r(v)\}$ (see Figure \ref{fig:twocases} (b)). 
Imagine we now build the following random graph $\tilde{G}_{uv}^{local} = (\tilde{V}, \tilde{E})$: The vertex set $\tilde{V}$ is simply ${B}_{uv}$. To construct the edge set $\tilde{E}$, first, add all edges in the clique spanned by nodes in $B^V_r(u)$ as well as edges in the clique spanned by nodes in $B^V_r(v)$ into $\tilde{E}$. Next, add edge $uv$ to $\tilde{E}$. Finally, insert each crossing edge $xy$ with $x\in B^V_r(u)$ and $y\in B^V_r(v)$ with probability $q$. 

On the other hand, consider the graph $\Ghat|_{{B}_{uv}}$, the induced subgraph of $\Ghat$ spanned by vertices in ${B}_{uv}$. We can imagine that the graph $\Ghat|_{B_{uv}}$ was produced by first taking the induced subgraph $\tG|_{{B}_{uv}}$, and then insert crossing edges $xy$ each with probability $q$. Since $uv$ is a \mybadE{}, by Definition \ref{def:goodedge}, we know that there are no edges between nodes in $B^V_r(u)$ and $B^V_r(v)$ in the random geometric graph $\tG$. In other words, edges in $\tG|_{B_{uv}}$ will be a subset of the two cliques spanned by $B^V_r(u)$ and $B^V_r(v)$, respectively. 
Hence we obtain: 
\begin{align}\label{eqn:insertiononlycaseBfirststep}
&\myprob\left[\Ghat|_{B_{uv}}~\text {has a}~uv\text{-clique of size} \ge \frac{\aK}{2} \right] 
~~\leq~~ \myprob\left[\tilde{G}_{uv}^{local} ~\text {has a}~uv\text{-clique of size} \ge \frac{\aK}{2}\right]
\end{align}    

Using a similar argument as in case (A) (the missing details can be found in Appendix \ref{appendix:detailscaseBinsertiononly}), we have that there exist constants $c_1^b,c_2^b,c_3^b >0$ which depend on the doubling constant $L$, the Besicovitch constant $\myBC$ and the regularity constant $\rho$ such that
\begin{align*}
\text{If} ~ q \leq ~ \min \left\{c_1^b, c^b_2 \cdot \left(\frac{1}{n}\right)^{c^b_3/\aK}\cdot \frac{\aK}{\mys n}\right\} \text{, then} ~ \myprob \left[\tilde{G}_{uv}^{local} ~\text {has a}~uv\text{-clique of size} \ge \frac{\aK}{2} \right] = O(n^{-3})
\end{align*}

Thus, combining this with Eqn. (\ref{eqn:insertiononlycaseBfirststep}), (\ref{eqn:caseAfinal}) and (\ref{eqn:mainresult}), there exist constants $c_1 = \min\{c_1^a, c_1^b\}$, $c_2 = \min\{c^{a}_2, c^{b}_2\}$ and $c_3 = \max\{c^{a}_1, c^{b}_3\}$ such that if $q$ satisfies conditions in Eqn. (\ref{eqn:qbound}), then 
\begin{align*}
\myprob\left[\Ghat ~\text{has a}~uv\text{-clique of size} \ge \aK\right] = O(n^{-3})
\end{align*}

Finally, by applying the union bound, this means:
\begin{align*}
\myprob\left[ \text{for every \mybadE{} } (u,v) \text{, } \Ghat ~\text{has a}~uv\text{-clique of size} \ge \aK\right]  = O(n^{-1})
\end{align*}
Thus with high probability, we have that for every \mybadE{} $(u,v)$, $\omega_{u,v}(\Ghat) < \aK$ as long as Eqn. (\ref{eqn:qbound}) holds. Since each $(u,v)$ with $d(u,v) > 3r$ is a \mybadE{}, the statement is also true for those edges.
This finishes the proof of Theorem \ref{thm:insertiononlybadedge}. 


\subsection{Edge clique numbers for the deletion-only case}
\label{subsec:deletion}

We now consider the deletion-only case, where we assume that the input graph $\Ghat = \Ghat^p$ is obtained by deleting each edge in the random geometric graph $G^* = G^*_\X(r)$ independently with probability $p$. 
For an edge $(u,v)$ in $\Ghat$, below we will give a lower bound on the edge-clique number $\omega_{u,v}(\Ghat)$. A simple observation is that for any edge $(u,v)$ in $\tG$, as  $d_X(u,v) \le r$, we have that $B_r(u) \cap B_r(v)$ must contain a metric ball of radius $r/2$ (say $B_{r/2}(z)$ centered at midpoint $z$ of a geodesic connecting $u$ to $v$ in $X$; see Figure \ref{fig:edgetypes} (a)). Thus by a similar argument as the proof of Claim \ref{claim:degreebound}, the number of points in the $r/2$-ball can be bounded from below w.h.p. 
Note that all points in a $r/2$-ball span a clique in the random geometric graph $G^*$. Since we then remove each edge from $G^*$ independently (to obtain $\Ghat$), to find a lower bound for $\omega_{u,v}(\Ghat)$, it suffices to consider the ``local'' subgraph of $\Ghat$ restricted within this $r/2$-ball $B_{r/2}(z)$. This local graph has the same behavior as the standard \myER{} random graph $G(N_{z}, 1-p)$, where $N_{z}$ is the number of points from $V$ within the ball $B_{r/2}(z)$. This eventually leads to the following result, whose proof is in 
Appendix \ref{appendix:subsec:combinedgoodclique}. 
Note that in the deletion-only case, all edges are \mygoodE{}s, so we only need to discuss the behavior of the edge clique number for \mygoodE{}s. 

\begin{theorem}\label{thm:combinedgoodclique}
Let $\tG=G_{\X}^*(r)$ be an $n$-node random geometric graph generated by $(X, d, \mu)$ where $\mu$ is an $L$-doubling probability measure. Assume \mydense{} holds. Let $\Ghat = \Ghat^{p}$ denote the final graph after deleting each edge in $\tG$ independently with probability $p$. Then, for any constant $p \in (0, 1)$, with high probability, 
we have $\omega_{u,v}(\Ghat) \geq \frac{2}{3}\log_{1/(1-p)}{sn}$ for all edges $(u,v)$ in $\Ghat$.
\end{theorem}

\myparagraph{Remark. } 
Note that with high probability, the edge clique number is $\Omega\left(\log_{1/(1-p)} (sn)\right)$ -- This is significantly smaller than the case when $p=0$, which is $\Omega(sn)$. 
On the other hand, this is inherited from the behavior of the \myER{} graph, where the large clique size is only logarithmic of the number of nodes for any constant insertion probability ($1-p$ in our case) smaller than $1$ (intuitively, the subgraph within the neighborhood of any \mygoodE{} is a subgraph of a certain \myER{} graph $G(csn, 1-p)$ for some constant $c$).


\subsection{Combined Case}
\label{subsec:combined}

In this section, we consider both the deletion and insertion. In other words, we consider the ER-perturbed random geometric graph $\Ghat$ generated via the model described in section \ref{sec:definition} that includes both edge-deletion probability $p$ and edge-insertion probability $q$. 
Our main results for the combined case are summarized in the following theorem. 
For the bound for \mygoodE{}s, it is easy to see that the same argument to bound the edge clique number for \mygoodE{}s in the deletion-only case works here, giving rise to the same lower bound. 

For \mybadE{}s, we can apply a similar strategy used in Section \ref{subsec:insertion} for the insertion-only case, where the main difference is that the ``local cliques'' (formed by only \mygoodE{}s) are now qualitatively smaller due to the deletion probability. The (somewhat repetitive) details can be found in Appendix \ref{appendix:thm:combinedbadclique}. 


\begin{theorem}[ER-perturbed random geometric graph]\label{thm:combinedbadclique}
Let $\tG=G_{\X}^*(r)$ be an $n$-node random geometric graph generated from $(\X, d, \mu)$ where $\mu$ is an $L$-doubling probability measure supported on $X$. Suppose {\sf \myassumption{}} holds. Let $\Ghat = \Ghat^{p,q}(r)$ denote the graph obtained by removing each edge in $\tG$ independently with constant probability $p \in (0, 1)$ and inserting each edge not in $\tG$ independently with probability $q$. There exist constants $c_1, c_2, c_3> 0$ which depend on the doubling constant $L$ of $\mu$, the \Bconst{} $\myBC(\X)$, and the regularity constant $\myRC$, such that the following holds for any $\aK = \aK(n)$ with $\aK\rightarrow \infty$ as $n \rightarrow \infty$ 
\begin{itemize}
\item[1.] W.h.p., for all \mygoodE{}s $(u,v)\in \Ghat$, $\omega_{u,v}(\Ghat) \geq \frac{2}{3}\log_{1/(1-p)}{sn}$.
\item[2.] W.h.p., for all \mybadE{}s $(u,v) \in \Ghat$, $\omega_{u,v}(\Ghat) < \aK$ as long as the insertion probability $q$ satisfies
\begin{align}\label{eqn:qboundcombined}
q \le \min\left\{c_1, ~~c_2 \cdot \left(\frac{1}{n}\right)^{c_3/\aK}\cdot \frac{\aK}{\mys n \sqrt{1-p}}\right\}
\end{align} 

\end{itemize}
\end{theorem}

\myparagraph{Remark. } 
For example, assume $sn = \Theta(\ln n)$. Then for a constant deletion probability $p\in (0, 1)$, w.h.p. the edge clique number for any \mygoodE{} is {\bf at least} $\Omega\left(\log_{1/(1-p)}{\mys n}\right) = \Omega(\ln \ln n)$. 
For any \mybadE{} $uv$, if the insertion probably $q = o\left({(\frac{1}{n})^{\frac{c_3}{\ln \ln n}}}\frac{\ln \ln n}{\ln n}\right)$, then its 
edge clique number is {\bf at most} $\aK = o\left(\log_{1/(1-p)}{\mys n}\right) = o(\ln \ln n)$ w.h.p. (Note that the quantity $(\frac{1}{n})^{\frac{c_3}{\ln \ln n}}$ is asymptoticly larger than $n^{-\eps}$ for any $\eps>0$; that is, $(\frac{1}{n})^{\frac{c_3}{\ln \ln n}} = \omega(n^{-\eps})$ for any $\eps > 0$). 

As $q$ decreases, the gap between the edge clique number for \mygoodE{}s and \mybadE{}s can be made larger and larger. 

Compared to the insertion-only case, it may seem that the condition on $q$ is too restrictive (recall that for the insertion only case we only require $q = o(1)$ to have a gap between edge clique number for \mygoodE{}s and \mybadE{}s). Intuitively, this is because: even for an \myER{} graph $G(n, q)$ with $q = {(\frac{1}{n})^{\frac{c_3}{\ln \ln n}}}\frac{\ln \ln n}{\ln n}$, its clique number is of order $\Theta(\ln \ln n)$ with high probability\footnote{Indeed, the upper bound can be easily derived by computing the expectation; and Lemma \ref{thm:ERclique} in the Appendix provides the lower bound.}. This clique size is already at the same scale as the bound of edge clique number for a \mygoodE{} in the deletion-only case. Intuitively, this now gets into a regime where the good/\mybadE{}s potentially have edge cliques of asymptotically similar sizes.


\section{Recover the shortest-path metric of $\tG(r)$}\label{sec:metricrecovery} 

In this section, we show an application in recovering the shortest-path metric structure of $G_{\X}^{*}(r)$ from an input observed graph $ \Ghat_{\X}^{p,q}(r)$. 
This problem is previously introduced in \cite{ParthasarathyST17}. 
Intuitively, assume that $G^* = G_\X^*(r)$ is the true graph of interests (which reflects the metric structure of $(X,d)$), but the observed graph is a ($p,q$)-perturbed version $\Ghat = \Ghat_{\X}^{p,q}(r)$ as described in Section \ref{sec:definition}. 
The goal is to recover the shortest-path metric of $G^*$ from its noisy observation $\Ghat$ with approximation guarantees. Note that due to the random insertion, two nodes could have significantly shorter path in $\Ghat$ than in $G^*$. 

Specifically, given two different metrics defined on the same space $(Y, d_1)$ and $(Y, d_2)$, we say that $d_1 \le \alpha \cdot d_2$ if for any two points $y_1, y_2 \in Y$, we have that $d_1(y_1, y_2) \le \alpha \cdot d_2(y_1, y_2)$. 
The metric $d_1$ is an \emph{$\alpha$-approximation} of $d_2$ if $\frac{1}{\alpha} \cdot d_2 \le d_1 \le \alpha \cdot d_2$ for $\alpha \ge 1$ and $\alpha=1$ means that $d_1 = d_2$. 

Let $d_G$ denote the shortest-path metric on graph $G$.
It was observed in \cite{ParthasarathyST17} that, roughly speaking, deletion (with $p$ smaller than a certain constant) does not distort the shortest-path metric of $G^*$ by more than a factor of $2$. 
Insertion however could change shortest-path distances significantly. The authors of \cite{ParthasarathyST17} then proposed a filtering process to remove some ``bad'' edges based on the so-called Jaccard index, and showed that after the Jaccard-filtering process, the shortest-path metric of the resulting graph $\tilde{G}$ $2$-approximates that of the true graph $G^*$ when the insertion probability  $q$ is small.  

We follow the same framework as \cite{ParthasarathyST17}, but change the filtering process to be one based on the edge clique number instead. This allows us to recover the shortest-path metric within constant factor for a much larger range of values of the insertion probability $q$, although we do need the extra \myRegular{} which is not needed in \cite{ParthasarathyST17}. (Note that it does not seem that the bound of \cite{ParthasarathyST17} can be improved even with this extra \myRegular{}). 
  
We now introduce our edge-clique based filtering process.  
\begin{description}\denselist
\item {\sf $\athreshold$-Clique filtering}: 
Given graph $\Ghat$, we construct another graph $\anotherG_\athreshold$ on the same vertex set as follows: For each edge $(u,v) \in E(\Ghat)$, we insert the edge $(u,v)$ into $E(\anotherG_\athreshold)$ if and only if $\omega_{u,v}(\Ghat) \ge \athreshold$. That is, $V(\anotherG_\athreshold) = V(\Ghat)$ and $E(\anotherG_\athreshold) := \left\{ (u,v) \in E(\Ghat) \mid \omega_{u,v}(\Ghat) \ge \athreshold \right\}$. 
\end{description}

A simple application of Theorem \ref{thm:combinedbadclique} (i) and (ii) gives the following two lemmas, respectively.
\begin{lemma}\label{lm:combinedkeepgoodedges}
Under the same setting as Theorem \ref{thm:combinedbadclique}, if $p \in (0, 1)$ and the filtering parameter $\athreshold$ satisfies $\athreshold < \frac{2}{3}\log_{1/(1-p)}{sn}$, then, with high probability, our $\athreshold$-Clique filtering process will not remove any \mygoodE{}s. 
\end{lemma}

\begin{lemma}\label{lm:combineddeletebadedges}
Under the same setting as Theorem \ref{thm:combinedbadclique}, there exist constants $c_1, c_2, c_3 > 0$ such that for constant $p \in (0,1)$, with high probability, a $\athreshold$-Clique filtering process deletes all \mybadE{}s, as long as 
$q \le \min \left\{ c_1, ~~c_2 \cdot (\frac{1}{n})^{c_3/\tau}\cdot \frac{\tau}{\mys n \sqrt{1-p}} ~\right\}.$
\end{lemma}

The following result can be proved by almost the same argument as that for Theorem 12 of \cite{ParthasarathyST17}, with the help of the lemmas above. For completeness, the proof is in Appendix \ref{appendix:thm:maincombined}. 
\begin{theorem}\label{thm:maincombined}
Let $G^* = G^*_{\X}(r)$ be an $n$-node random geometric graph generated from $(\X, d, \mu)$ where $\mu$ is an $L$-doubling probability measure supported on $X$. Assume {\sf \myassumption{}} holds. Suppose $\Ghat = \Ghat^{p,q}(r)$ is the graph obtained after deleting each edge in $\tG$ independently with constant probability $p$ and inserting each edge not in $\tG$ independently with probability $q$.
Let $\anotherG_\athreshold$ denote the resulting graph after $\athreshold$-Clique filtering. 
 Then there exist constants $c_1, c_2, c_3> 0$ which depend on the doubling constant $L$ of $\mu$, the \Bconst{} $\myBC(\X)$, and the regularity constant $\myRC$, such that 
if $p \in (0, c_0)$, $\tau \le \frac{2}{3}\log_{1/(1-p)}{sn}$, and 
\begin{align*}
q &\le c_2 \cdot \left(\frac{1}{n}\right)^{c_3/\tau}\cdot \frac{\tau}{\mys n \sqrt{1-p}} ~~~\left(= \min\left\{~c_1, ~~c_2 \cdot \left(\frac{1}{n}\right)^{c_3/\tau}\cdot \frac{\tau}{\mys n \sqrt{1-p}}~\right\}\right), 
\end{align*}
then, with high probability, the shortest-path metric $d_{\anotherG_\athreshold}$ is a 3-approximation of the shortest-path metric $d_\tG$ of $\tG$. 

However, if the deletion probability $p = 0$, then we have w.h.p. that $d_{\anotherG_\athreshold}$ is a $3$-approximation of $d_\tG$ as long as $\tau < \frac{sn}{4}$, and $q \le \min\left\{~c_1, ~~c_2 \cdot \left(\frac{1}{n}\right)^{c_3/\tau}\cdot \frac{\tau}{\mys n}~\right\}$. 
\end{theorem}

\myparagraph{Remark. } 
To give some example of the above theorem, first consider the insertion-only case (i.e, the deletion probability $p = 0$), which is a case of independent interest. 
In this case, if we choose $\tau = \ln n$ and assume that $sn > 4\tau$, then 
w.h.p. we can recover the shortest-path metric within a factor of 3 as long as
$q \le c \frac{\ln n}{sn}$ for some constant $c > 0$. 
If $sn = \Theta(\ln n)$ (but $sn > 4\tau = 4\ln n)$, then $q$ is only required to be smaller than a (sufficiently small) constant. 
If $sn = \ln^a n$ for some $a > 1$, then we require that $q \le \frac{c}{\ln^{a-1} n}$. 
In constrast, the work of \cite{ParthasarathyST17} requires that $q = o(s)$, which is $q = o(\frac{\ln^c n}{n})$ if $sn= \ln^a n$ with $a \ge 1$. The gap (ratio) between these two bounds is nearly a factor of $n$. 

For a \emph{constant} deletion probability $p \in (0, c_0)$, our clique filtering process still requires a much larger range of insertion probability $q$ compared to what's required in \cite{ParthasarathyST17}, although the gap is much smaller than the case for $p=0$. 
For example, assume $sn = \Theta(\ln n)$. Then if we choose the filtering parameter to be $\tau = \sqrt{\ln \ln n}$, then we can recover $d_{\tG}$ approximately as long as the insertion probability $q = o\left( \left(\frac{1}{n}\right)^{\frac{c_3}{\sqrt{\ln \ln n}}}\frac{\sqrt{\ln \ln n}}{\ln n}\right)$. This is still much larger than the $q$ required in \cite{ParthasarathyST17}, which is $q = o(s) = o\left(\frac{\ln n}{n}\right)$. In fact, $(\frac{1}{n})^{\frac{c_3}{\sqrt{\ln \ln n}}}\frac{\sqrt{\ln \ln n}}{\ln n}$ is asymptoticly larger than $\frac{1}{n^\eps}$ for any $\eps > 0$. 
However, we do point out that the Jaccard-filtering process in \cite{ParthasarathyST17} is algorithmically much simpler and faster, and can be done in $O(n^2)$ time, while the clique-filtering requires the computation of edge-clique numbers, which is computationally expensive.

\myparagraph{Acknowledgements.} 
This work is partially supported by National Science Foundation under grants DMS-1547357, CCF-1740761 and RI-1815697.

\bibliography{clique_number_RGG_under_ER_perturbation}

\appendix

\section{Proof of Claim \ref{claim:degreebound}}
\label{appendix:claim:nbhd}
\begin{proof}
For a fixed vertex $v\in V$, let $n_v$ be the number of points in $(V-\{v\}) \cap \aball_r(v)$. 
The expectation of $n_v$ is $ (n - 1) \cdot \mu\left(\aball_r(v)\right) \ge (n-1) \cdot \mu\left(B_{\frac{r}{2}}(v)\right)\ge s(n-1)$. 
By the Chernoff bound, we thus have that 
\begin{align*}
\myprob\left[n_v < \frac{sn}{4}\right] < \myprob\left[n_v < \frac{s(n-1)}{3}\right] &\le \myprob\left[n_v < \frac{1}{3}(n-1) \mu\left(\aball_r(v)\right)\right]  \le e^{-\frac{\left(\frac{2}{3}\right)^{2}}{2}(n-1)\mu(\aball_r(v))} \le  n^{-\frac{8}{3}}
\nonumber\\ 
\end{align*}

It then follows from the union bound that the probability that all $n$ vertices in $V$ have more than $sn/4$ neighbors is at least  $1 - n \cdot n^{-\frac{8}{3}} = 1 - n^{-5/3}$. 
\end{proof}

\section{The missing proofs in Section \ref{subsec:insertion}}
\label{appendix:thm:insertiononly}

\subsection{Proof of Theorem \ref{thm:insertiononlygoodedge}}
\label{appendix:thm:insertiononlygoodedge}

For each \mygoodE{} $(u,v)$, observe that $B_r(u) \cap B_r(v)$ contains at least one metric ball of radius $r/2$ (say $B_{r/2}(z)$ with $z$ being the mid-point of a geodesic connecting $u$ to $v$ in $X$, see Figure \ref{fig:edgetypes} (a)). And all the points in an $r/2-$ball span a clique in $\tG$ ($r-$neighborhood graph). Then by an argument similar to the proof of Claim \ref{claim:degreebound}, we have that with probability at least  $1 - n^{-\frac{2}{3}}$, the number of points in all of $O(n^2)$ number $r/2$-balls centered at some mid-point of the geodesics between all pair of nodes $u, v \in V$ is at least $\mys n/4$. 
Hence with probability at least  $1 - n^{-\frac{2}{3}}$, for all \mygoodE{} $(u,v)$ in $\Ghat$, $\omega_{u,v}(\Ghat) \geq \mys n/4$.

\subsection{The proof of Claim \ref{lm:r/2ballupperbound}}
\label{appendix:claim:r/2ballupperbound}
\begin{proof}
For a fixed vertex $v \in V$, let $n_{v,r/2}$ be the number of points in $(V-\{v\})\cap B_{r/2}(v)$. By the definition of random geometric graph, we know that $n_{v,r/2}$ is subject to binomial distribution $Bin\left(n-1, \mu\left(B_{r/2}(v)\right)\right)$. The expectation of $n_{v,r/2}$ is $(n-1)\mu(B_{r/2}(v)) \leq \myRC \mys n$. Also note that $(n-1)\mu\left(B_{r/2}(v)\right) \geq (n-1)s\geq 12 \ln n$. By applying the Chernoff bound, we thus have that
\begin{align*}
\myprob\left[n_{v,r/2} \geq \frac{5}{2}\myRC \mys n\right] \leq \myprob\left[n_{v,r/2} \geq \frac{5}{2}(n-1)\mu\left(B_{r/2}(v)\right)\right]\leq e^{-\frac{1}{3}\left(\frac{3}{2}\right)(n-1)\mu\left(B_{r/2}(v)\right)} \leq n^{-6}
\end{align*}

Finally, by applying the union bound, we know that with probability at least $1-n\cdot n^{-6} = 1- n^{-5}$, $\forall v \in G^{*}$, there are at most $\frac{5}{2}\myRC \mys n + 1 < 3\myRC \mys n$ points in the geodesic ball $B_{r/2}(v)$.
\end{proof}

\subsection{The proof of Lemma \ref{lem:insertioncaseA}}
\label{appendix:lem:insertioncaseA}
\begin{proof}
We first pick $c(\myBC)$ to be a positive constant such that $ \floor*{\frac{c(\myBC)N_{\max}}{2|\Lambda|}}-2 \leq N_{\max}$. Then, since $k = \floor*{\frac{\aK}{2|\Lambda|}} -2$, it is easy to see that for any $\aK \leq c(\myBC)N_{\max} = c(\myBC)3\rho \mys n$, we have $k \leq N_{\max}$. Pick $c =  c(\beta) 3\rho$. Then, $\aK \leq c \mys n$ implies $k \leq N_{\max}$.

To estimate the summation on the right hand side of Eqn. (\ref{expectationX}), we consider the quantity $x_{\max} := \max\limits_{i}\{x_i\}$. We first enumerate all the possible cases of $(x_1, x_2, \cdots, x_m)$ when $x_{\max}$ is fixed, and then vary the value of $x_{\max}$.  

Set $h(y) = \max\limits_{x_{\max}=y}\left\{\sum\limits_{i=1}^{m}x_i^2\right\}$ for $y \geq \ceil*{\frac{k}{m}}$. It is the maximum value of $\sum\limits_{i=1}^{m}x_i^2$ under the constraint $x_{\max}=y$. Without loss of generality, we assume $x_1=y$ and $y \geq x_2 \geq x_3 \geq \cdots \geq x_m \geq 0$. We argue that $\argmax\limits_{x_{\max}=y}\left\{\sum\limits_{i=1}^{m}x_i^2\right\} = \{y, y, \cdots, y, k-ry, 0, \cdots, 0\}$, that is $x_1 = x_2 = \cdots = x_r = y, x_{r+1}=k-ry$ where $r = \floor*{\frac{k}{y}}$. 

To show this, we first consider $x_2$: if $x_2 = y$, then consider $x_3$; otherwise, $x_2 < y$, then we search for the largest index $j$ such that $x_j >0$. Note the fact that if $x \geq y>0$, then $(x+1)^2+(y-1)^2 = x^2 + y^2 + 2(x-y) + 2>x^2+y^2$. So if we increase $x_2$ by $1$ and decrease $x_j$ by $1$, we will enlarge $\sum\limits_{i=1}^{m}x_i^2$. After we update $x_2 = x_2 + 1$, $x_j = x_j - 1$, we still get a decreasing sequence $x_1\geq x_2 \geq \cdots \geq x_m \geq 0$. If we still have $x_2 < y$, 
then we repeat the same procedure above (by increasing $x_2$ and decreasing $x_j$ where $j$ is the largest index such that $x_j>0$). 
We repeat this process until $x_2 = y$ or $x_1 + x_2 = k$. If it is the former case (i.e, $x_2=y$), then we consider $x_3$ and so on. Finally, we will get the sequence $x_1 = x_2 = \cdots = x_r = y, x_{r+1}=k-ry$ where $r = \floor*{\frac{k}{y}}$ as claimed, and this setting maximizes $\sum\limits_{i=1}^{m}x_i^2$. 

Next we claim that $h(y+1) > h(y)$. The reason is similar to the above. We update the sequence $x_1 = x_2 = \cdots = x_r = y, x_{r+1}=k-ry$ (which corresponding to $h(y)$) from $x_1$: we increase $x_1$ by $1$; search the largest index $s$ such that $x_{s} > 0$ and decrease $x_{s}$ by $1$. And then consider $x_2$ and so on and so forth. This process won't stop until $x_1 = x_2 = \cdots = x_q =y+1$ and $x_{q+1} = k - q(y+1)$ with $q = \floor*{\frac{k}{y+1}}$. Thus $h(y+1)>h(y)$. 

By enumerating all the possible values of $x_{\max}$, we split Eqn. (\ref{expectationX}) into three parts as follows (corresponding to $x_{\max} = k, x_{\max} \in \left[ \ceil*{\frac{k+1}{2}}, k-1\right] \text{~and~}x_{\max} \in \left[\ceil*{\frac{k}{m}}, \ceil*{\frac{k+1}{2}}-1\right]$) (see the remarks after this equation for how the inequality is derived); 
\begin{align}\label{eqn:mainestimate}
&q^{2k}\sum\limits_{\substack{x_1 + x_2 + \cdots + x_m = k \\x_i \geq 0}} \binom{N_{\max}}{x_1} \binom{N_{\max}}{x_2}\cdots \binom{N_{\max}}{x_m}q^{(k^2-\sum_{i=1}^{m}x_i^2)/2} \nonumber\\
&\leq ~~q^{2k}\binom{N_{\max}}{k}m ~~+~~  q^{2k}\sum\limits_{x_{\max} =\ceil*{\frac{k+1}{2}}}^{k-1}\Bigg(\binom{m}{1} \binom{N_{\max}}{x_{\max}} \sum\limits_{\substack{y_1+\cdots+y_{m-1} = k - x_{\max}\\ x_{\max} \geq y_i \geq 0}}\binom{N_{\max}}{y_1} \cdots \nonumber\\
&\binom{N_{\max}}{y_{m-1}}q^{x_{\max}(k-x_{\max})} \Bigg) ~~+~~ \binom{mN_{\max}}{k}q^{\frac{(k-1)^2}{4} + 2k}\nonumber\\
\end{align}

\myparagraph{Remark. } The first term on the right hand side of Eqn. (\ref{eqn:mainestimate}) comes from the fact that if $x_{\max} = k$, then all the possible cases for $(x_1, x_2, \cdots, x_m)$ are $(k,0,0,\cdots,0),(0,k,0,\cdots,0),\cdots, (0,\cdots,0,k)$, and there are $m$ cases all together. 
For each case, the value of each term in the summation is $\binom{N_{\max}}{k}$, giving rise to the first term in Eqn. (\ref{eqn:mainestimate}).

The third term on the right hand side of Eqn. (\ref{eqn:mainestimate}) can be derived as follows. First, observe that 
\begin{align*}
&\sum\limits_{x_{\max} = \ceil*{\frac{k}{m}}}^{\ceil*{\frac{k+1}{2}} - 1}\Bigg(\sum\limits_{\substack{x_1 + x_2 + \cdots + x_m = k \\x_i \geq 0, \max_{i}\{x_i\} = x_{\max}}} \binom{N_{\max}}{x_1} \binom{N_{\max}}{x_2}\cdots \binom{N_{\max}}{x_m}\Bigg)\\
\leq &\sum\limits_{\substack{x_1 + x_2 + \cdots + x_m = k \\x_i \geq 0}} \binom{N_{\max}}{x_1} \binom{N_{\max}}{x_2}\cdots \binom{N_{\max}}{x_m} = \binom{mN_{\max}}{k} . 
\end{align*}
On the other hand, as $x_{\max} \leq \ceil*{\frac{k+1}{2}} - 1 = \ceil*{\frac{k-1}{2}}$, we have:
\begin{align*}
\frac{k^2-\sum_{i=1}^{m}x_i^2}{2} \geq \frac{k^2 -h(x_{\max})}{2} \geq \frac{k^2 -h(\ceil*{\frac{k-1}{2}})}{2} \geq \frac{(k-1)^2}{4}, 
\end{align*}
where the second inequality uses the fact that $h(y)$ is an increasing function, and the last inequality comes from that $h(\ceil*{\frac{k-1}{2}}) \le (\ceil*{\frac{k-1}{2}})^2 + (\ceil*{\frac{k-1}{2}})^2 + 1 \le k^2/4 + k^2/4 + 1 = k^2/2 + 1$. 
\vspace{5mm}

In what remains, it suffices to estimate all the three terms on the right hand side of Eqn. (\ref{eqn:mainestimate}). 

\myparagraph{The first term of Eqn. (\ref{eqn:mainestimate}):} According to the assumptions in Eqn. (\ref{eqn:qbound1}), we know $q \leq \left(\dfrac{k!}{n^{\epsilon}N_{\max}^km}\right)^{1/2k}$. Thus, for the first term of Eqn. (\ref{eqn:mainestimate}), we have:
\begin{align}\label{eqn:firsttermestimate}
q^{2k}\binom{N_{\max}}{k}m ~~\leq~~ \left(\dfrac{k!}{n^{\epsilon}N_{\max}^km}\right) \dfrac{N_{\max}^k}{k!}m ~~=~~ \frac{1}{n^{\epsilon}}
\end{align}

\myparagraph{The second term of Eqn. (\ref{eqn:mainestimate}):}For the second term of Eqn. (\ref{eqn:mainestimate}), we relax the constraint $x_{\max} \geq y_i \geq 0$ to $y_i \geq 0$. Thus, we have:
\begin{align}\label{eqn:boxandballtrick}
\sum\limits_{\substack{y_1+\cdots+y_{m-1} = k - x_{\max}\\x_{\max} ~~\geq~~ y_i \geq 0}}\binom{N_{\max}}{y_1}\cdots &\binom{N_{\max}}{y_{m-1}} \leq \sum\limits_{\substack{y_1+\cdots+y_{m-1} = k - x_{\max}\\y_i \geq 0}}\binom{N_{\max}}{y_1}\cdots \binom{N_{\max}}{y_{m-1}}\nonumber\\
& =~~ \binom{(m-1)N_{\max}}{k-x_{\max}} ~~\leq~~ \dfrac{(m-1)^{k-x_{\max}}N_{\max}^{k-x_{\max}}}{(k-x_{\max})!} 
\end{align}

Now apply (\ref{eqn:boxandballtrick}) to the second term of (\ref{eqn:mainestimate}), we have (starting from the second line, we replace $x_{max}$ to be $j$ for simplicity):
\begin{align}
&q^{2k}\sum\limits_{x_{\max} =\ceil*{\frac{k+1}{2}}}^{k-1}\Bigg(\binom{m}{1} \binom{N_{\max}}{x_{\max}} \sum\limits_{\substack{y_1+\cdots+y_{m-1} = k - x_{\max}\\x_{\max} \geq y_i \geq 0}}\binom{N_{\max}}{y_1} \cdots \binom{N_{\max}}{y_{m-1}}q^{x_{\max}(k-x_{\max})} \Bigg) \nonumber\\
&\leq \sum\limits_{j =\ceil*{\frac{k+1}{2}}}^{k-1}\Bigg(m \frac{(N_{\max})^j}{j!}q^{2k + j(k-j)}\dfrac{(m-1)^{k-j}N_{\max}^{k-j}}{(k- j)!}\Bigg)\nonumber\\
& < \sum\limits_{j =\ceil*{\frac{k+1}{2}}}^{k-1}\Bigg(m^{k-j+1} N_{\max}^{k} \binom{k}{k-j} \frac{1}{k!} q^{2k + j(k-j)} \Bigg)\nonumber\\
& < \sum\limits_{j =\ceil*{\frac{k+1}{2}}}^{k-1}\Bigg(m^{k-j+1} N_{\max}^{k} \frac{k^{k-j}}{(k-j)!} \frac{1}{k!} q^{2k + j(k-j)} \Bigg)\label{eqn:secondtermmain}
\end{align}

Since $q \leq \left(\dfrac{k!}{k^2n^{\epsilon}N_{\max}^km^2}\right)^{1/k}$ by Eqn. (\ref{eqn:qbound1}), for each $j$ satisfying $\ceil*{\frac{k+1}{2}} \leq j \leq k-1$, we have:
\begin{align}
&m^{k-j+1} N_{\max}^{k} \frac{k^{k-j}}{(k-j)!} \frac{1}{k!} q^{2k + j(k-j)} \nonumber\\
& \leq~~ m^{k-j+1} N_{\max}^{k} \frac{k^{k-j}}{(k-j)!} \frac{1}{k!} \dfrac{k!}{k^2n^{\epsilon}N_{\max}^km^2} \left(\dfrac{k!}{k^2n^{\epsilon}N_{\max}^km^2}\right)^{\frac{2k+j(k-j)}{k}-1} \nonumber\\
& \leq~~ m^{k-j-1}k^{k-j-1}\left(\dfrac{k!}{k^2n^{\epsilon}N_{\max}^km^2}\right)^{\frac{k+j(k-j)}{k}}\frac{1}{kn^{\epsilon}}\nonumber\\
& \leq~~ m^{k-j-1}k^{k-j-1}\left(\dfrac{k!}{k^2n^{\epsilon}N_{\max}^km^2}\right)^{\frac{k-j-1}{2}}\frac{1}{kn^{\epsilon}}\label{eqn:xmaxbound}\\
& =~~ \left(\dfrac{k!}{n^{\epsilon}N_{\max}^k}\right)^{\frac{k-j-1}{2}} \frac{1}{kn^{\epsilon}}\nonumber\\
& \leq~~ \frac{1}{kn^{\epsilon}}\label{eqn:uselowerbound}
\end{align}

Eqn. (\ref{eqn:xmaxbound}) comes from two facts: 1) $k \leq N_{\max}$ (and thus the term $\dfrac{k!}{k^2n^{\epsilon}N_{\max}^km^2}< 1$) and 2) by tedious by elementary calculation, we can show that $\frac{k+j(k-j)}{k}\geq \frac{k-j-1}{2}$ when $\ceil*{\frac{k+1}{2}} \leq j \leq k-1$. 
Eqn. (\ref{eqn:uselowerbound}) holds since $k \leq N_{\max}$ (and thus $\dfrac{k!}{n^{\epsilon}N_{\max}^k} < 1$). 

\myparagraph{The third term of Eqn. (\ref{eqn:mainestimate}):} For the third term of (\ref{eqn:mainestimate}), we just plug in the condition $q \leq \left(\dfrac{k!}{n^{\epsilon}m^kN_{\max}^{k}}\right)^{4/k^2}$:
\begin{align}\label{eqn:2ndterm}
\binom{mN_{\max}}{k} q^{\frac{(k-1)^2}{4}+2k} &~~\leq~~ \frac{(mN_{\max})^k}{k!}q^{\frac{(k-1)^2}{4}+2k} \nonumber\\
&~~\leq~~ \frac{(mN_{\max})^k}{k!} \dfrac{k!}{n^{\epsilon}m^kN_{\max}^{k}}\left(\dfrac{k!}{n^{\epsilon}m^kN_{\max}^{k}}\right)^{\frac{4}{k^2}\left(\frac{(k-1)^2}{4}+2k\right)-1} ~~\leq~~ \frac{1}{n^{\epsilon}}
\end{align}
where the last inequality holds as $\dfrac{k!}{n^{\epsilon}m^kN_{\max}^{k}} < 1$. 

Finally, combining (\ref{eqn:firsttermestimate}), (\ref{eqn:uselowerbound}) and (\ref{eqn:2ndterm}), we have:
\begin{align*}
E[\aX \mid \aF] \leq \frac{1}{n^{\epsilon}} + \frac{k}{2} \cdot \frac{1}{kn^{\epsilon}}  + \frac{1}{n^{\epsilon}} = \frac{5}{2n^{\epsilon}}
\end{align*}
This proves Lemma \ref{lem:insertioncaseA}. 
\end{proof}

\subsection{Existences of constants $c_2^a$ and $c_3^a$}
\label{appendix:constantsc2c3}

We claim that there exist constants $c^a_2$ and $c^a_3$ (which depend on the doubling constant $L$ of $\mu$, the \Bconst{} $\myBC$, and the regularity constant $\myRC$), such that if 
\begin{align*}
q \le c^a_2 \cdot \left(\frac{1}{n}\right)^{c^a_3/\aK}\cdot \frac{\aK}{\mys n},
\end{align*}
then the conditions in Eqn. (\ref{eqn:qbound1}) will hold. 
We prove this by elementary calculation below, where we will use the Stirling's approximation $k! > \sqrt{2\pi}k^k e^{-k}$ and the fact that $k \leq N_{\max} = 3\rho sn$ ($\aK \leq c \mys n$ implies this due to the choice of $c$ in the proof of Lemma \ref{lem:insertioncaseA}) and $m \le n$ (where recall that $m$ is the size of the number of clusters in the clique-decomposition of $P_i$). 
\begin{align*}
&\left(\frac{k!}{n^{\epsilon} N_{\max}^k m} \right)^{\frac{1}{2k}} > \left(\frac{\sqrt{2\pi} k^k e^{-k}}{n^{1+\epsilon}} \right)^{\frac{1}{2k}} \left(\frac{1}{3\rho sn} \right)^{\frac{1}{2}} = \left(e^{-\frac{1}{2}} (2\pi)^{\frac{1}{4k}}\right) \left(\frac{1}{n} \right)^{\frac{1+\epsilon}{2k}} \left(\frac{k}{3\rho sn} \right)^{\frac{1}{2}}\\
&\left(\frac{k!}{k^2n^{\epsilon} N_{\max}^k m^2} \right)^{\frac{1}{k}} > \left(\frac{\sqrt{2\pi} k^k e^{-k}}{k^2} \right)^{\frac{1}{k}}\left(\frac{1}{n} \right)^{\frac{2+\epsilon}{k}}\left(\frac{1}{3\rho sn} \right) = \left(e^{-1} (2\pi)^{\frac{1}{4k}} k^{-\frac{2}{k}}\right) \left(\frac{1}{n} \right)^{\frac{2+\epsilon}{k}}\left(\frac{k}{3\rho sn} \right)\\
&\left(\frac{k!}{n^{\epsilon} m^k N_{\max}^k} \right)^{\frac{4}{k^2}} > \left(\frac{1}{3\rho sn} \right)^{\frac{4}{k}} (\sqrt{2\pi} k^k e^{-k})^{\frac{4}{k^2}}\left(\frac{1}{n}\right)^{\frac{4(k+\epsilon)}{k^2}} = \left((2\pi)^{\frac{2}{k^2}} e^{-\frac{4}{k}}\right) \left(\frac{1}{n}\right)^{\frac{4(k+\epsilon)}{k^2}}\left(\frac{k}{3\rho sn} \right)^{\frac{4}{k}}
\end{align*}
Thus, by comparing the exponents of each term, we know that if $q \leq \frac{1}{3\rho e}\left(\frac{1}{n}\right)^{\frac{4+\epsilon}{k}}\left(\frac{k}{sn}\right)$, then the condition on $q$ holds. 
Finally, since $k = \floor*{\frac{\aK}{2|\Lambda|}}-2$, easy to see there exists constants $c_2^{a}, c_3^{a}$ such that the constraint $q \leq c_2^a \left(\frac{1}{n}\right)^{c_3^a/\aK}\frac{\aK}{sn}$ implies the condition.

\subsection{The missing details in case (B) of Theorem \ref{thm:insertiononlybadedge}}\label{appendix:detailscaseBinsertiononly}
Denote by $\aH$ the event that ``for every $v \in V$, the ball $B_{r}(v)\cap V$ contains at most $3L\rho \mys n$ points'', and $\aH^c$ is its complement. 
By an argument similar to that of Claim \ref{lm:r/2ballupperbound}, we have that $\myprob[\aH^c] \le n^{-5}$. Set $N_u := |B^V_r(u)|$ and $N_v := |B^V_r(v)|$. Let $\tilde{k} := \floor*{\frac{\aK}{2}} - 2$. For every set $S$ of $(\tilde{k}+2)$ vertices in $\tilde{G}_{uv}^{local}$, let $A_S$ be the event ``\emph{$S$ is a $uv$-clique in $\tilde{G}_{uv}^{local}$}'' and $\aY_S$ its indicator random variable. Set
\begin{align*}
\aY = \sum\limits_{|S|=\tilde{k}+2} \aY_S. 
\end{align*}
Then $\aY$ is the number of $uv$-cliques of size $(\tilde{k}+2)$ in $\tilde{G}_{uv}^{local}$. Linearity of expectation gives:
\begin{align}\label{eqn:expectationI}
\myE[\aY \mid \aH]~~ =~~ \sum\limits_{|S|=\tilde{k}+2}\myE[\aY_S\mid \aH] ~~=~~ \sum\limits_{\substack{x_1 + x_2 = \tilde{k} \\ 0 \leq x_1 \leq N_u - 1, 0\le x_2 \le N_v - 1}} \binom{N_u-1}{x_1} \binom{N_v-1}{x_2}q^{(x_1+1)(x_2+1)-1}
\end{align}

To estimate this quantity, we first prove the following result:

\begin{lemma}\label{lem:insertioncaseB}
For any constant $\epsilon>0$, 
we have that $\myE[\aY \mid \aH] = O(n^{-\epsilon})$ as long as the following condition on $q$ holds: 
\begin{align}\label{eqn:qbound2}
q ~~\leq~~ \min \left\{\left(\dfrac{\tilde{k}!}{{\tilde{k}}^2n^{\epsilon}(N_u+N_v)^{\tilde{k}}}\right)^{1/\tilde{k}},~~ \left(\frac{\tilde{k}!}{n^{\epsilon}(N_u+N_v)^{\tilde{k}}}\right)^{16/\tilde{k}^2}\right\}. 
\end{align}

Specifically, setting $\epsilon=3$ (a case which we will use later), we have $\myE[\aY\mid \aH] = O(n^{-3})$.
\end{lemma}

The proof of this technical result can be found in Appendix \ref{appendix:lem:insertioncaseB}. 

Note that if event $\aH$ is true, then $N_u+N_v \leq 6L\myRC \mys n$. 
In this case, there exist two constants $c^{b}_2$ and $c^{b}_3$ which depend on the doubling constant $L$ of $\mu$, the \Bconst{} $\myBC$, and the regularity constant $\myRC$, 
such that if $\aK \leq 12L\rho \mys n$ and 
\begin{align*}
q \leq  c^b_2 \cdot \left(\frac{1}{n}\right)^{c^b_3/\aK}\cdot \frac{\aK}{\mys n},
\end{align*}
then the conditions in Eqn. (\ref{eqn:qbound2}) will hold (the simple proof of this can be found in Appendix \ref{appendix:caseBconstantsc2c3}). 

On the other hand, we have
\begin{align*}
\myprob\left[ \tilde{G}_{uv}^{local} ~\text {has a}~uv\text{-clique of size} \ge \frac{\aK}{2} \right] &~=  \myprob[\aY > 0] \\
&~=~ \myprob[\aY > 0 \mid \aH] \cdot \myprob[\aH] + \myprob[\aY > 0 \mid \aH^c] \cdot \myprob[\aH^c] \\
&~\leq~ \myprob[\aY > 0 \mid \aH] + \myprob[\aH^c]\\
&~\leq~ \myE[\aY \mid \aH] + n^{-5}. 
\end{align*}
Thus, by Lemma (\ref{lem:insertioncaseB}), we know that 
\begin{align}\label{eqn:caseBalmostfinal}
\text{If} ~\aK \leq 6L\rho \mys n~\text{and}~ q \leq ~ c^b_2 \cdot \left(\frac{1}{n}\right)^{c^b_3/\aK}\cdot \frac{\aK}{\mys n} \text{, then} ~ \myprob \left[\tilde{G}_{uv}^{local}~\text {has a}~uv\text{-clique of size} \ge \frac{\aK}{2} \right] = O(n^{-3})
\end{align}

Finally, suppose $\aK > \aK_1 = 12L\rho \mys n$. Set
\begin{align*}
c_1^b = c^b_2 \cdot \left(\frac{1}{n}\right)^{c^b_3/(12L\rho \ln n)}\cdot \frac{\aK_1}{\mys n} \leq c^b_2 \cdot \left(\frac{1}{n}\right)^{c^b_3/\aK_1}\cdot \frac{\aK_1}{\mys n}
\end{align*}
where the inequality holds as by \myassumption{} $\mys n > \ln n$. Plugging in $\aK_1 = 12L\rho \mys n$ to the definition of $c_1^b$, it is then easy to see that $c_1^b$ is a positive constant. Using Eqn. (\ref{eqn:caseBalmostfinal}), we know that if $q \leq c_1^b$ and $\aK > \aK_1 = 12L \rho \mys n$, then
\begin{align*}
\myprob\left[ \tilde{G}_{uv}^{local} ~\text {has a}~uv\text{-clique of size} \ge \frac{\aK}{2} \right] \leq \myprob\left[ \tilde{G}_{uv}^{local} ~\text {has a}~uv\text{-clique of size} \ge \frac{\aK_1}{2} \right] = O(n^{-3})
\end{align*}

Combining this with the discussion above and applying Lemma \ref{lem:insertioncaseB}, we have
\begin{align*}
\text{If} ~ q \leq ~ \min \left\{c_1^b, c^b_2 \cdot \left(\frac{1}{n}\right)^{c^b_3/\aK}\cdot \frac{\aK}{\mys n}\right\} \text{, then} ~ \myprob \left[\tilde{G}_{uv}^{local}~\text {has a}~uv\text{-clique of size} \ge \frac{\aK}{2} \right] = O(n^{-3})
\end{align*}

\subsection{The proof of Lemma \ref{lem:insertioncaseB}}
\label{appendix:lem:insertioncaseB}
\begin{proof}
Easy to see that if $\aK > 2(N_u + N_v)$, then $\tilde{k} > N_u + N_v - 2$ which implies that the summation on the right hand side of Eqn. (\ref{eqn:expectationI}) is $0$.
Now let's focus on the case when $\aK \leq 2(N_u + N_v)$. In this case, we have $\tilde{k} \leq (N_u -1 )+ (N_v - 1) < N_u + N_v$. Note that the right hand side of (\ref{eqn:expectationI}) can be bounded from above by:

\begin{align}
&\sum\limits_{\substack{x_1 + x_2 = \tilde{k} \\ 0 \leq x_1 \leq N_u - 1, 0 \leq x_2 \leq N_v - 1}} \binom{N_u-1}{x_1} \binom{N_v-1}{x_2}q^{(x_1+1)(x_2+1)-1} ~~\leq~~ q^{\tilde{k}}\sum\limits_{i=0}^{\tilde{k}}\binom{N_u}{i}\binom{N_v}{\tilde{k}-i}q^{i(\tilde{k}-i)}\nonumber\\
&\leq~~ q^{\tilde{k}}\left(\sum\limits_{i=0}^{\floor*{\frac{\tilde{k}}{4}}}\left[\binom{N_u}{i}\binom{N_v}{\tilde{k}-i} ~~+~~ \binom{N_u}{\tilde{k}-i}\binom{N_v}{i}\right]q^{i(\tilde{k}-i)}\right) ~~+~~ \binom{N_u+N_v}{\tilde{k}}q^{\tilde{k}+\frac{\tilde{k}^2}{16}}\label{eqn:localballestimate}
\end{align}

Eqn. (\ref{eqn:localballestimate}) is due to the fact that when $\floor*{\frac{\tilde{k}}{4}} + 1 \leq i \leq \tilde{k} - \floor*{\frac{\tilde{k}}{4}} - 1 $, we have $i(\tilde{k}-i) \geq \left(\floor*{\frac{\tilde{k}}{4}} + 1\right)\left(\floor*{\frac{\tilde{k}}{4}} + 1\right) \geq \frac{\tilde{k}^2}{16}$. Now it suffices to estimate the two terms on the right hand side of Eqn. (\ref{eqn:localballestimate}).

\myparagraph{The first term of Eqn. (\ref{eqn:localballestimate}):} For the first term of (\ref{eqn:localballestimate}), we have the following estimate:

\begin{align}
\bigg[\binom{N_u}{i}\binom{N_v}{\tilde{k}-i} + \binom{N_u}{\tilde{k}-i}\binom{N_v}{i}\bigg]q^{\tilde{k}+i(\tilde{k}-i)} &~~\leq~~ \dfrac{N_u^iN_v^{\tilde{k}-i} + N_u^{\tilde{k}-i}N_v^i}{i!(\tilde{k}-i)!}q^{\tilde{k}}q^{i(\tilde{k}-i)} \nonumber\\
&~~\leq~~ \dfrac{(N_u+N_v)^{\tilde{k}}}{i!(\tilde{k}-i)!}q^{\tilde{k}}q^{i(\tilde{k}-i)}\nonumber
\end{align}
By plugging in the condition $q \leq \left(\dfrac{\tilde{k}!}{\tilde{k}^2n^{\epsilon}(N_u+N_v)^{\tilde{k}}}\right)^{1/\tilde{k}}$, we have:

\begin{align}
\dfrac{(N_u+N_v)^{\tilde{k}}}{i!(\tilde{k}-i)!}q^{\tilde{k}}q^{i(\tilde{k}-i)} \leq \dfrac{(N_u+N_v)^{\tilde{k}}}{i!(\tilde{k}-i)!} \frac{\tilde{k}!}{\tilde{k}^2n^{\epsilon}(N_u+N_v)^{\tilde{k}}}q^{i(\tilde{k}-i)} = \dfrac{\tilde{k}!}{i!(\tilde{k}-i)!}q^{i(\tilde{k}-i)}\frac{1}{\tilde{k}^2n^{\epsilon}}\nonumber\\\nonumber
\end{align}

For $i = 0$, we have $\dfrac{\tilde{k}!}{i!(\tilde{k}-i)!}q^{i(\tilde{k}-i)}\frac{1}{\tilde{k}^2n^{\epsilon}} =  \frac{1}{\tilde{k}^2n^{\epsilon}}$. For $i \geq 1$, note that $1 \leq i \leq \floor*{\frac{\tilde{k}}{4}}$ implies $\frac{i(\tilde{k}-i)}{\tilde{k}} \geq \frac{i}{2}$. Thus, we have:
\begin{align*}
\dfrac{\tilde{k}!}{i!(\tilde{k}-i)!}q^{i(\tilde{k}-i)}\frac{1}{\tilde{k}^2n^{\epsilon}} &~~\leq~~ \frac{\tilde{k}^i}{i!}\left(\dfrac{\tilde{k}!}{\tilde{k}^2n^{\epsilon}(N_u+N_v)^{\tilde{k}}}\right)^{\frac{i(\tilde{k}-i)}{\tilde{k}}}\frac{1}{\tilde{k}^2n^{\epsilon}} \\
&~~\leq~~ \frac{\tilde{k}^i}{i!}\left(\dfrac{\tilde{k}!}{\tilde{k}^2n^{\epsilon}(N_u+N_v)^{\tilde{k}}}\right)^{\frac{i}{2}}\frac{1}{\tilde{k}^2n^{\epsilon}} ~~\leq~~ \left(\dfrac{\tilde{k}!}{n^{\epsilon}(N_u+N_v)^{\tilde{k}}}\right)^{\frac{i}{2}}\frac{1}{\tilde{k}^2n^{\epsilon}} ~~\leq~~ \frac{1}{\tilde{k}^2n^{\epsilon}}\\
\end{align*}
The last two inequalities hold since $\tilde{k} \leq  N_u + N_v$. Therefore,
\begin{align}
q^{\tilde{k}}\sum\limits_{i=0}^{\floor*{\frac{\tilde{k}}{4}}}\left[\binom{N_u}{i}\binom{N_v}{\tilde{k}-i} + \binom{N_u}{\tilde{k}-i}\binom{N_v}{i}\right]q^{i(\tilde{k}-i)} ~~\leq~~ \frac{\tilde{k}}{4}\frac{1}{\tilde{k}^2n^{\epsilon}} ~~=~~ \frac{1}{4\tilde{k}n^{\epsilon}}\label{eqn:localballfirstterm}
\end{align}

\myparagraph{The second term of Eqn. (\ref{eqn:localballestimate}):} For the second term of (\ref{eqn:localballestimate}), directly plugging in the condition $q \leq \left(\frac{\tilde{k}!}{(N_u+N_v)^{\tilde{k}}n^{\epsilon}}\right)^{16/\tilde{k}^2}$, we have:
\begin{align}
\binom{N_u+N_v}{\tilde{k}}q^{\tilde{k}+\frac{\tilde{k}^2}{16}} ~~\leq~~ \frac{(N_u+N_v)^{\tilde{k}}}{\tilde{k}!}\frac{\tilde{k}!}{(N_u+N_v)^{\tilde{k}}n^{\epsilon}}\left(\frac{\tilde{k}!}{(N_u+N_v)^{\tilde{k}}n^{\epsilon}}\right)^{\frac{16}{\tilde{k}^2}\left(\tilde{k}+\frac{\tilde{k}^2}{16}\right)-1} ~~\leq~~ \frac{1}{n^{\epsilon}}\label{eqn:localballsecondterm}
\end{align} 

Finally, combining (\ref{eqn:localballfirstterm}) and (\ref{eqn:localballsecondterm}), we have:
\begin{align*}
\myE[\aY \mid \aH] \leq \frac{1}{4\tilde{k}n^{\epsilon}} + \frac{1}{n^{\epsilon}} < \frac{2}{n^{\epsilon}}
\end{align*}
This finishes the proof of Lemma \ref{lem:insertioncaseB}. 
\end{proof}

\subsection{Existences of constants $c_2^b$ and $c_3^b$}\label{appendix:caseBconstantsc2c3}

Note that as event $\aH$ holds, we have $N_u + N_v \leq 6L\rho sn$. Also note that if $\aK \leq 12 L \rho \mys n$, then $\tilde{k} \leq 6 L \rho \mys n$. Hence
\begin{align*}
&\left(\frac{\tilde{k}!}{\tilde{k}^2 n^{\epsilon} (N_u + N_v)^{\tilde{k}}} \right)^{\frac{1}{\tilde{k}}} > \left(\frac{\sqrt{2\pi}\tilde{k}^{\tilde{k}} e^{-\tilde{k}}}{\tilde{k}^2 n^{\epsilon}} \right)^{\frac{1}{\tilde{k}}}\frac{1}{6L \rho sn} = \left( (2\pi)^{\frac{1}{2\tilde{k}}}\tilde{k}^{-\frac{2}{\tilde{k}}} e^{-1}\right) \left( \frac{1}{n} \right)^{\frac{\epsilon}{\tilde{k}}} \left(\frac{\tilde{k}}{6L \rho sn} \right)\\
&\left(\frac{\tilde{k}!}{n^{\epsilon} (N_u + N_v)^{\tilde{k}}} \right)^{\frac{16}{\tilde{k}^2}} > \left(\frac{1}{n} \right)^{\frac{16\epsilon}{\tilde{k}^2}} \left(\sqrt{2\pi}\tilde{k}^{\tilde{k}} e^{-\tilde{k}}\right)^{\frac{16}{\tilde{k}^2}} \left(\frac{1}{6L \rho sn} \right)^{\frac{16}{\tilde{k}}} = \left(2\pi \right)^{\frac{8}{\tilde{k}^2}}e^{-\frac{16}{\tilde{k}}} \left(\frac{1}{n} \right)^{\frac{16\epsilon}{\tilde{k}^2}} \left(\frac{\tilde{k}}{6L \rho sn} \right)^{\frac{16}{\tilde{k}}}
\end{align*}

Finally, observe that $(2\pi)^{\frac{8}{k^2}} > 1$, $e^{-\frac{16}{k}} > \frac{1}{e}$.
It then follows that there exists constants $c_2^b,c_3^b$ such that $c_2^b \left(\frac{1}{n} \right)^{c_3^b/\aK}\frac{\aK}{sn}$ is smaller than the last term in the right hand side of each equation above. Hence $\tilde{k} \leq 6 L \rho \mys n$ and $q \leq c_2^b \left(\frac{1}{n} \right)^{c_3^b/\aK}\frac{\aK}{sn}$ implies the condition on $q$ as in Eqn. (\ref{eqn:qbound2}).  

\section{Edge clique numbers for the deletion-only case}
\label{appendix:subsec:deletion}

In this section, we will give a lower bound on $\omega_{u,v}(\Ghat)$ for the deletion-only case. 
On the high level, we first prove the following lemma for an \myER{} random graph, via an application of Janson's Inequality \cite{alon2016probabilistic}. The proof is in Section \ref{appendix:subsec:ERclique}. 


\begin{lemma}\label{thm:ERclique}
Suppose $G = G(N, \bar{p})$ is an \myER{} random graph with $\bar{p} \in \left(\left(\frac{1}{N}\right)^{\frac{1}{10}} , \left(\frac{1}{N}\right)^{\frac{1}{\sqrt[64]{N}}}\right)$. Let $k = \floor*{\log_{1/\bar{p}}{N}}$, then
\begin{align*}
\myprob[\omega(G) < k] < e^{-N^{3/2}}
\end{align*} 
where $\omega(G)$ is the clique number of graph $G$.
\end{lemma}

\myparagraph{Remark. }  
One can easily verify that $(\frac{1}{N})^{\frac{1}{10}}$ is very close to $0$ and $(\frac{1}{N})^{\frac{1}{\sqrt[64]{N}}}$ is very close to $1$ as $N$ goes to infinity. Hence the range $\bar{p} \in \left(\left(\frac{1}{N}\right)^{\frac{1}{10}} , \left(\frac{1}{N}\right)^{\frac{1}{\sqrt[64]{N}}}\right)$ is broader (significantly more relaxed) than requiring that $\bar{p}$ is a constant between $(0,1)$. 
Hence, while not pursued in the present paper, it is possible to show that Theorem \ref{thm:combinedgoodclique} holds for a larger range of $p$. 


\subsection{Proof of Lemma \ref{thm:ERclique}}
\label{appendix:subsec:ERclique}

We first introduce the following two Janson's inequalities \cite{alon2016probabilistic}
\begin{theorem}[Janson's Inequality]\label{thm:JansonInequality}
Let $\Omega$ be a finite universal set, and let $R$ be a random subset of $\Omega$ given by $\myprob[r \in R] = p_{r}$, with these events being mutually independent over $r\in \Omega$. Let $\{A_i\}_{i\in I}$ be subsets of $\Omega$, and $I$ a finite index set. Let $B_{i}$ be the event $A_{i} \subseteq R$ and $\bar{B_i}$ be its complement. Let $\aX_i$ be the indicator random variable for $B_i$, and $\aX := \sum_{i\in I} \aX_i$ the number of $A_{i} \subseteq R$. For $i,j\in I$, we write $i \sim j$ if $i \neq j$ and $A_i \cap A_j \neq \emptyset$. We define $\Delta := \sum\limits_{i \sim j}\myprob[B_i \wedge B_j].$
Here, the sum is over ordered pairs. Finally, we set $\zeta : = E[\aX] = \sum\limits_{i\in I}\myprob[B_i].$ Then
\begin{align}
\myprob\left[\bigwedge_{i\in I}\bar{B_i}\right] \leq e^{-\zeta + \Delta /2}.
\end{align}
\end{theorem}

\begin{theorem}[The Extended Janson's Inequality]\label{thm:extendedJanson}
Under the assumptions of Theorem \ref{thm:JansonInequality} and a further assumption that $\Delta \geq \zeta$, then
\begin{align}
\myprob\left[\bigwedge_{i\in I}\bar{B_i}\right] \leq e^{-\zeta^{2}/2\Delta}
\end{align}
\end{theorem}

\myparagraph{Proof of Lemma \ref{thm:ERclique}.} 
Consider all the $k-$set $A_i$ of vertices in $G(N,\bar{p})$, let $B_i$ be the event ``$A_i$ is a clique in $G(N, \bar{p})$'' and $\aX_i$ its indicator random variable. Let $I$ be the finite index set enumerating all the $k-$sets in $G(N, \bar{p})$. Set
\begin{align*}
\aX = \sum\limits_{i\in I}\aX_i ;
\end{align*}
thus $\aX$ is the number of $k-$cliques in $G(N,\bar{p})$. Linearity of Expectation gives:
\begin{align}
\zeta = \myE[\aX] = \sum\limits_{i \in I}\myE[\aX_i]= \binom{N}{k}\bar{p}^{\binom{k}{2}}
\end{align}
For any fixed index $i$, we set
\begin{align}
\Delta^{*}_i:=\sum\limits_{j \in I:j \sim i}\myprob[B_j | B_i] = \sum\limits_{\ell=2}^{k-1}\binom{k}{\ell}\bar{p}^{\binom{k}{2}-\binom{\ell}{2}}\binom{N-k}{k-\ell}
\end{align}
It's easy to check that the following holds independent of $i$ (the details can be found in \cite{alon2016probabilistic}):
\begin{align*}
\Delta^*:= \Delta^{*}_i=\frac{\Delta}{\myE[X]} = \frac{\Delta}{\zeta}, 
\end{align*}
where $\Delta$, as before, is $\Delta := \sum\limits_{i,j\in I; i \sim j}\myprob[B_i \wedge B_j]$. 

The conditions $k = \floor*{\log_{1/\bar{p}}{N}}$ and $(\frac{1}{N})^{\frac{1}{10}} < \bar{p} < (\frac{1}{N})^{\frac{1}{\sqrt[64]{N}}}$ imply that $k \leq \log_{1/\bar{p}}{N} < k+1$, $ 10 < k < \sqrt[64]{N}$ and $\frac{\ln k}{\ln N} < \frac{1}{64}$. Note that for sufficiently large $N$, we have $N - k > \frac{N}{2}$. Thus

\begin{align*}
\zeta = \binom{N}{k}(\bar{p})^{\binom{k}{2}} \geq \frac{(N-k)^k}{k^k}(\bar{p})^{\frac{k(k-1)}{2}} > \left(\frac{N}{2k}\right)^{k} N ^{-\frac{k}{2}} = N^{\frac{k}{2}\left(1-\frac{2\ln{(2k)}}{\ln N}\right)} > N^{\frac{k}{2}\cdot \frac{1}{2}} > 2N^{\frac{3}{2}}\\
\end{align*}

If $\Delta \leq \zeta$, then by Theorem \ref{thm:JansonInequality}, we have
\begin{align*}
\myprob[\omega(G) < k] = \myprob[\aX = 0] = \myprob\left[\bigwedge_{i\in I}\bar{B_i}\right] \leq e^{-\zeta + \Delta/2} \leq e^{-\zeta/2} < e^{-N^{3/2}}
\end{align*}

Otherwise, if $\Delta \geq \zeta$, we want to apply Theorem \ref{thm:extendedJanson} to this case. It suffices to estimate the term $\frac{\zeta^2}{\Delta} = \frac{\zeta}{\Delta^*}$. In what follows, we estimate $\frac{\Delta^*}{\zeta}$ instead.
\begin{align}\label{eqn:delta*/mu}
\frac{\Delta^*}{\zeta} = \dfrac{\sum\limits_{\ell=2}^{k-1}\binom{k}{\ell}\binom{N-k}{k-\ell}(\bar{p})^{\binom{k}{2}-\binom{\ell}{2}}}{\binom{N}{k}(\bar{p})^{\binom{k}{2}}} =\sum\limits_{\ell=2}^{k-1}\dfrac{\binom{k}{\ell}\binom{N-k}{k-\ell}}{\binom{N}{k}}(\bar{p})^{-\binom{\ell}{2}}
\end{align}

Now set 
\begin{align*}
g(\ell) := \frac{\binom{k}{\ell}\binom{N-k}{k-\ell}}{\binom{N}{k}}(\bar{p})^{-\binom{\ell}{2}},
\end{align*}
thus the right hand side of Eqn. (\ref{eqn:delta*/mu}) equals $\sum_{\ell = 2}^{k-1} g(\ell)$. 
Below we will show that $g(\ell) \leq \max\{g(2),g(k-1)\}$ and will then get an upper bound for $g(2)$ and $g(k-1)$ respectively in order to upper bound $g(\ell)$. Note that for $\ell \in [2,k-1]$, we have
\begin{align}
g(\ell) & ~~=~~ \frac{\binom{k}{\ell}\binom{N-k}{k-\ell}}{\binom{N}{k}}(\bar{p})^{-\binom{\ell}{2}} ~~\leq~~ \frac{\binom{k}{\ell}\frac{(N-k)^{k-\ell}}{(k-\ell)!}}{\frac{(N-k)^k}{k!}}N^{\frac{1}{k}\frac{\ell(\ell-1)}{2}} ~~\leq~~ \frac{\binom{k}{\ell}k^\ell}{(N-k)^\ell}N^{\frac{1}{k}\frac{\ell(\ell-1)}{2}} \nonumber\\
& ~~\leq~~ \frac{(\frac{ek}{\ell})^\ell k^\ell}{(N-k)^\ell}N^{\frac{1}{k}\frac{\ell(\ell-1)}{2}} ~~=~~ N^{-\frac{\ell\ln{\left(\frac{\ell(N-k)}{ek^2}\right)}}{\ln{N}} + \frac{1}{k}\frac{\ell(\ell-1)}{2}}
\end{align}

Set 
\begin{align*}
h(\ell) := -\frac{\ell\ln{\left(\frac{\ell(N-k)}{ek^2}\right)}}{\ln{N}} + \frac{1}{k}\frac{\ell(\ell-1)}{2}
\end{align*}
and thus $g(\ell) = N^{h(\ell)}$. 
We claim that for any $\ell \in [2,k-1]$, we have $h(\ell) \leq \max \{h(2), h(k-1)\}$. 
Indeed, we can prove this by the following elementary calculations:

The derivative of $h(\ell)$ with respect to $\ell$ is
\begin{align*}
h'(\ell) = -\frac{\ln{\ell} + \ln{(N-k) - 2\ln{k}}}{\ln{N}} + \frac{2\ell-1}{2k}.
\end{align*} 
Next calculate its second derivative: 
\begin{align*}
h''(\ell) = -\frac{1}{\ln{N}}\frac{1}{\ell} + \frac{1}{k}.
\end{align*}
Note that $\ell_0 = \frac{k}{\ln{N}}$ is the only solution of $h''(\ell)=0$. Easy to check that $\ell_0 \leq k-1$. Therefore, we have the following two cases: 
\begin{description}\denselist
\item[case-a] 
$\ell_0 < 2$, then $h'(\ell)$ is strictly increasing on $\ell\in[2,k-1]$.
\item[case-b]
$\ell_0 \in [2, k-1]$, then $h'(\ell)$ is strictly decreasing on $[2, \ell_0]$ and strictly increasing on $[\ell_0, k-1]$. 
\end{description}

\begin{figure}[h]
\centering
\begin{tabular}{ccccccc}
\includegraphics[height=1.7cm]{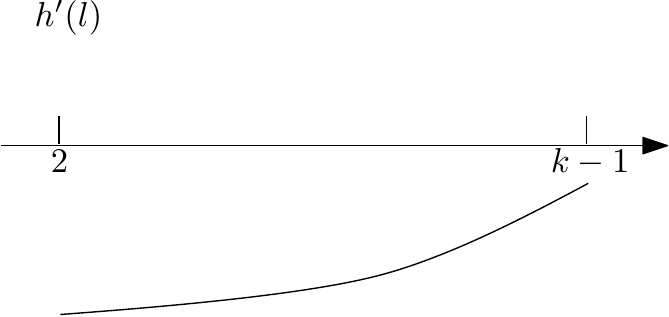} &\hspace*{0.05in}  & 
\includegraphics[height=1.7cm]{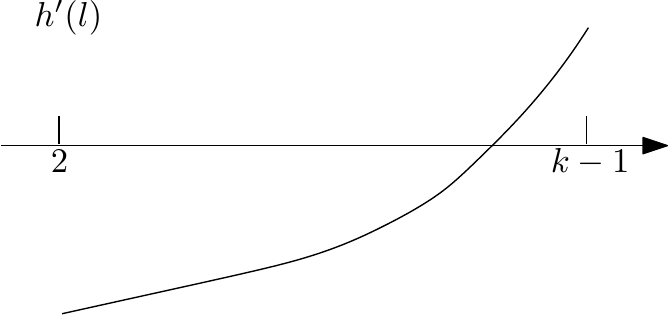} &\hspace*{0.05in}  & \includegraphics[height=1.7cm]{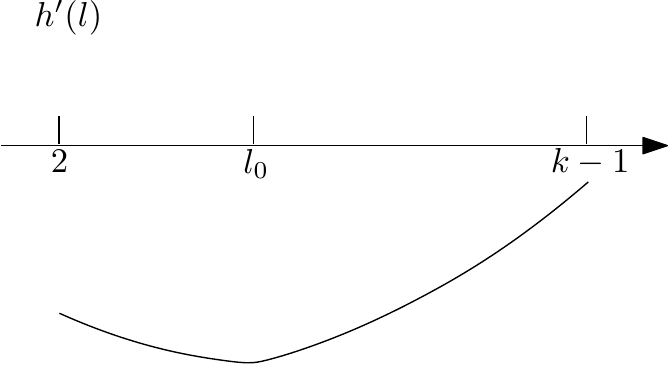} &\hspace*{0.05in}  & \includegraphics[height=1.7cm]{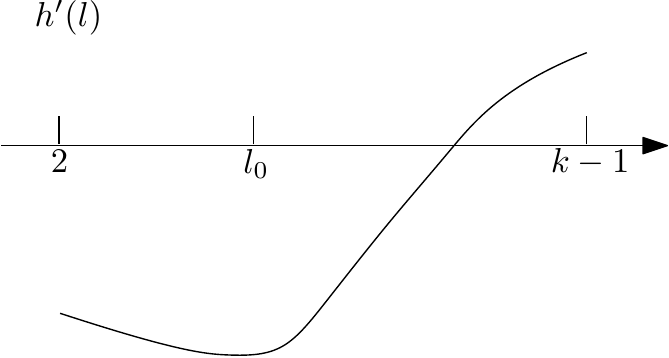}\\
\textbf{case-a}(1) &  & \hspace*{0in}\textbf{case-a}(2) &  & \hspace*{0in}\textbf{case-b}(1) &  & \hspace*{0in}\textbf{case-b}(2)
\end{tabular}
\vspace*{-0.15in}

\caption{All possible graphs of $h'(l)$. }
\label{fig:h_derivatives}
\end{figure}

Note that 
\begin{align*}
h'(2) < -\frac{\ln{2}+\ln{(N/2)} - 2\ln{k}}{\ln{N}} + \frac{3}{2k} = -1 +\frac{2\ln{k}}{\ln{N}} + \frac{3}{2k} < -1 + \frac{1}{32} + \frac{3}{20} < 0.
\end{align*}
Thus, by checking all possible graphs of $h'(l)$((see Figure \ref{fig:h_derivatives})), it's easy to see that in all cases, either $h(\ell$ is monotonically decreasing within range $\ell \in [2, k-1]$, or it first monotonically decreasing and then monotonically increasing within this range. In other words, the maximum is always achieved at one of the end point; that is, $\max_{\ell\in [2,k-1]}h(\ell) = \max \{h(2), h(k-1)\}$. 

Routine calculations show that 
\begin{align*}
h(2) < - \frac{2\left[\ln{(2\frac{N}{2}) - 1 - 2\ln k}\right]}{\ln N} + \frac{1}{k} = -2 + \frac{1}{k} + \frac{2}{\ln N}  + \frac{4\ln k}{\ln N} < -\frac{7}{4}
\end{align*}
and 
\begin{align*}
h(k-1) &~~<~~ -\frac{(k-1)\left[\ln{(\frac{N}{2}) -1 -\ln k - \ln \frac{k}{k-1}}\right]}{\ln N} + \frac{k^2-3k +2}{2k} \\
&~~<~~ \left[\frac{k^2-3k +2}{2k} - (k-1)\right] + \frac{k(1+ \ln 2)}{\ln N} + \frac{k\ln k }{\ln N} + \frac{k\ln \left(\frac{k}{k-1}\right)}{\ln N}\\
&~~<~~ -\frac{k}{2} - \frac{1}{2} + \frac{1}{10} + \frac{k}{8} + \frac{k}{64} + \frac{k}{64} < -\frac{7}{4}
\end{align*}
 Thus, $\forall \ell \in [2, k-1]$, we have $g(\ell) < N^{-\frac{7}{4}}$.

Hence, $\sum\limits_{\ell=2}^{k-1}g(\ell)  < k \cdot N^{-\frac{7}{4}}< \frac{1}{2} N^{-\frac{3}{2}}$. The second inequality is due to $k < \sqrt[64]{N} < \frac{\sqrt{N}}{2}$. 
Then by Eqn. (\ref{eqn:delta*/mu}), we have $\frac{\Delta^*}{\zeta} < \frac{1}{2} N^{-\frac{3}{2}}$, which implies $\frac{\zeta^2}{\Delta} = \frac{\zeta}{\Delta^*} > 2N^{\frac{3}{2}}$. Therefore, by Theorem \ref{thm:extendedJanson}, we have
\begin{align*}
\myprob[\omega(G) < k] = \myprob[\aX = 0] = \myprob\left[\bigwedge_{i\in I}\bar{B_i}\right] \leq e^{-\frac{\zeta^2}{2\Delta}} < e^{-N^{3/2}}.
\end{align*}
\subsection{Proof of Theorem \ref{thm:combinedgoodclique}}
\label{appendix:subsec:combinedgoodclique}
\begin{proof}
Using the argument in the proof of Claim \ref{claim:degreebound}, we know that for a fixed \mygoodE{} $(u,v)$ (i.e. $d(u,v) \leq r$), with probability $1 - n^{-\frac{8}{3}}$, the geodesic ball $B_{r/2}(z)$ ($z$ is the mid-point of a geodesic connecting $u$ to $v$ in $X$) contains at least $(\mys n/4)$ points. Note that all points in a $r/2$-ball form a clique in $r-$neighborhood graph. Since we remove each edge independently, in order to estimate $\omega_{u,v}(\Ghat)$ from below, it suffices to consider the ``local'' graph spanned by nodes in this $r/2$-ball. Note that this ``local'' graph have the same behavior as the standard \myER{} random graph $G_{uv}^{loc} := G(N_{z}, 1-p)$, where $N_{z}$ denotes the number of points falling in the ball $B_{r/2}(z)$. 

Furthermore, it is easy to see that, if $N_{z} \geq sn/4$, then for any {\it constant} $p \in (0,1)$, one can always find a sufficiently large $n$ such that $1-p \in \left((\frac{1}{N_{z}})^{\frac{1}{10}}, (\frac{1}{N_{z}})^{\frac{1}{\sqrt[64]{N_{z}}}}\right)$. 

Now, we are ready to apply Lemma \ref{thm:ERclique} to those ``local'' graphs: Note that $\bar{p}$ in Lemma \ref{thm:ERclique} will be set to be $1-p$, and $N$ will be set to be $N_{z}$.
\begin{align*}
\myprob\left[\omega (G_{uv}^{loc}) < k\mid N_{z} \geq \frac{sn}{4}\right] < e^{-(\frac{sn}{4})^{3/2}} \leq e^{-(3\ln{n})^{3/2}} < n^{-(\ln{n})^{1/2}}
\end{align*}
where $k = \floor*{\log_{1/(1-p)}{N_{z}}}$.

Hence for $k' = \frac{2}{3}\log_{1/(1-p)}{sn}$, we have that 
\begin{align*}
\myprob\left[\omega (G_{uv}^{loc}) < k'\mid N_{z} \geq \frac{sn}{4}\right]  < \myprob\left[\omega (G_{uv}^{loc}) < k \mid N_{z} \geq \frac{sn}{4}\right]  < n^{-(\ln{n})^{1/2}} .
\end{align*}
By the law of total probability, we know that
\begin{align*}
\myprob\left[\omega(G_{uv}^{loc}) < k' \mid d(u,v)\leq r\right] & < \myprob\left[\omega (G_{uv}^{loc}) < k' \mid N_{z} \geq \frac{sn}{4}\right] + \myprob\left[N_{z} < \frac{sn}{4}\right] \\
&< n^{-(\ln{n})^{1/2}} + n^{-\frac{8}{3}}
\end{align*} 

Applying the union bound, we have
\begin{align*}
\myprob\left[\bigwedge\limits_{u,v \in V; d(u,v) \leq r}\omega (G_{uv}^{loc}) \geq k'\right] &= 1 - \myprob\left[\exists u,v \in V~\text{with}~d(u,v) \le r~ \text{s.t.}~ \omega (G_{uv}^{loc}) < k'\right] \\
&\geq 1- \frac{1}{2}n^2\myprob\left[\omega(G_{uv}^{loc}) < k' \mid  d(u,v)\leq r\right] \\
&\geq 1 - \frac{1}{2}n^{-(\ln{n})^{1/2}+2} - \frac{1}{2}n^{-\frac{2}{3}}
\end{align*}

Thus, with high probability, for each \mygoodE{} $(u,v)$, we have $\omega_{u,v}(\Ghat) \geq k' = \frac{2}{3}\log_{1/(1-p)}{sn}$.
\end{proof}

\section{The proof of Theorem \ref{thm:combinedbadclique}}
\label{appendix:thm:combinedbadclique}
\begin{proof}
First, for part-(1) of Theorem \ref{thm:combinedbadclique}, note that as such a perturbed graph can have more edges than the deletion-only case, Theorem \ref{thm:combinedgoodclique} immediately implies this statement. 

In what follows, we prove part-(2) of the theorem, to upper bound the edge clique number of a \mybadE{} $(u,v) \in \Ghat$. The proof here is very similar to the proof of Theorem \ref{thm:insertiononlybadedge}, and we mainly need to adjust some technical details to include the deletion probability $p$. Given the similarity to the proof of Theorem \ref{thm:insertiononlybadedge}, we will use the same notations here and only provide main steps. 

Specifically, let $(u, v)$ be a \mybadE{} in $\Ghat$. Let $A_{uv}$, $\tilde{A}_{uv}$, and $B_{uv}$ as defined as in the proof of  Theorem \ref{thm:insertiononlybadedge}. 

First, by the pigeonhole principle and union bound, we have:
\begin{align}\label{eqn:mainresultcombined}
&\myprob\left[ \Ghat ~\text{has a}~uv\text{-clique of size} \ge \aK\right] \nonumber \\
\le &\myprob\left[ \Ghat|_{\tilde{A}_{uv}}~\text {has a}~uv\text{-clique of size} \ge \frac{\aK}{2}\right] + \myprob\left[ \Ghat|_{B_{uv}}~\text {has a}~uv\text{-clique of size} \ge \frac{\aK}{2}\right] 
\end{align}

\myparagraph{Case (A): bounding the first term in Eqn. (\ref{eqn:mainresultcombined}).} 
The arguments here proceed in the same way as in Case (A) of the proof of Theorem \ref{thm:insertiononlybadedge}, except for the  conditional expectation of $\aX$: Specifically, now $C_j^{(i)}$'s no longer span cliques but each $C_j^{(i)}$ spans an \myER{} random graphs with insertion probability $1-p$. Thus, we have 
\begin{align}
\myE[\aX \mid \aF] ~~=~~ &\sum\limits_{|S|=k+2}\myE[\aX_S \mid \aF] ~~=~~ q^{2k}\sum\limits_{\substack{x_1 + x_2 + \cdots + x_m = k \\ 0 \leq x_i \leq N_i}} \binom{N_1}{x_1}\cdots \binom{N_m}{x_m}q^{\frac{k^2-\sum_{i=1}^{m}x_i^2}{2}}(1-p)^{\sum_{i=1}^{m}\binom{x_i}{2}}\nonumber\\
&=~~ \left(\frac{q}{1-p}\right)^{2k}\sum\limits_{\substack{x_1 + x_2 + \cdots + x_m = k \\ 0 \leq x_i \leq N_i}} \binom{N_1}{x_1}\cdots \binom{N_m}{x_m}\left(\frac{q}{1-p}\right)^{\frac{k^2-\sum_{i=1}^{m}x_i^2}{2}}(1-p)^{\frac{k^2+3k}{2}}\nonumber\\
& \leq \left[\left(\frac{q}{1-p}\right)^{2k}\sum\limits_{\substack{x_1 + x_2 + \cdots + x_m = k \\ 0 \leq x_i \leq N_{\max}}} \binom{N_{\max}}{x_1}\cdots \binom{N_{\max}}{x_m}\left(\frac{q}{1-p}\right)^{\frac{k^2-\sum_{i=1}^{m}x_i^2}{2}}\right](1-p)^{\frac{k^2+3k}{2}}
\label{eqn:expectationXcombined}
\end{align}
  
\begin{lemma}\label{lem:combinedcaseA}
There exists a constant $c$ depending on $\myBC$ and $\rho$ such that for any constant $\epsilon>0$, if $\aK \leq c \mys n$ and $p,q$ satisfy: 
\begin{align}\label{eqn:combinedqbound1}
\frac{q}{1-p} ~~\leq~~ \min \Bigg\{& \left(\dfrac{k!}{n^{\epsilon}N_{\max}^km}\right)^{1/2k}\left(\frac{1}{1-p}\right)^{\frac{k^2+3k}{2}\frac{1}{2k}},~~ \left(\dfrac{k!}{k^2n^{\epsilon}N_{\max}^km^2}\right)^{1/k}\left(\frac{1}{1-p}\right)^{\frac{k^2+3k}{2}\frac{1}{2k+\frac{k^2}{4}}}~~, \nonumber\\
&\left(\dfrac{k!}{n^{\epsilon}m^kN_{\max}^{k}}\right)^{\frac{4}{k^2}}\left(\frac{1}{1-p}\right)^{\frac{k^2+3k}{2}\frac{1}{2k+\frac{(k-1)^2}{4}}}\Bigg\}
\end{align}
then we have that $\myE(\aX \mid \aF) = O(n^{-\epsilon})$. Specifically, we can set $\epsilon = 3$ (the choice will be necessary later to apply union bound) and obtain $\myE[\aX \mid \aF] = O(n^{-3})$. 
\end{lemma}

The proof of the above lemma follows almost exactly the argument for Lemma \ref{lem:insertioncaseA}: Indeed, roughly speaking, (1) we will now use $\frac{q}{1-p}$ for the role of ``$q$" as in (\ref{expectationX}); and (2) compared to Eqn. (\ref{eqn:qbound1}) in Lemma \ref{lem:insertioncaseA}, the extra terms related to $\left( \frac{1}{1-p}\right)$ in the right hand side of Eqn. (\ref{eqn:combinedqbound1}) is needed to compensate the extra term $(1-p)^{\frac{k^2+3k}{2}}$ in the right hand side of Eqn. (\ref{eqn:expectationXcombined}) (as compared to Eqn. (\ref{expectationX}). Other than these two points, the proof proceeds in the same as as that for Lemma \ref{lem:insertioncaseA}. 
We thus omit the tedious details. 

Furthermore, it is easy to check that as $p\in (0, 1)$, we have that 
\begin{align*}
\left(\frac{1}{1-p}\right)^{\frac{3}{2}} < \min \left\{\left(\frac{1}{1-p}\right)^{\frac{k^2+3k}{2}\frac{1}{2k}}, \left(\frac{1}{1-p}\right)^{\frac{k^2+3k}{2}\frac{1}{2k+\frac{k^2}{4}}}, \left(\frac{1}{1-p}\right)^{\frac{k^2+3k}{2}\frac{1}{2k + \frac{(k-1)^2}{4}}}    \right\}
\end{align*}

It then follows from an almost identical argument from Appendix \ref{appendix:constantsc2c3} that there exists $c_2^a, c_3^a$ such that  if $\aK \le csn$ and $q \le c_2^a\cdot \left(\frac{1}{n} \right)^{c_3^a/\aK}\cdot \frac{\aK}{\mys n\sqrt{1-p}}$, then the conditions in Eqn. (\ref{eqn:combinedqbound1}) hold. 
By a similar argument as in the proof of (case-A) for Theorem \ref{thm:insertiononlybadedge}, we thus obtain that: 
there exist constants $c_1^a, c_2^a, c_3^a$ such that 
\begin{align*}
\text{If } q \leq \min\left\{c_1^a, c_2^a\cdot \left(\frac{1}{n} \right)^{c_3^a/\aK}\cdot \frac{\aK}{\mys n\sqrt{1-p}} \right\}, ~\text{then}~\myprob\left[ \Ghat|_{\tilde{A}_{uv}}~\text {has a}~uv\text{-clique of size} \ge \frac{\aK}{2}\right] = O(n^{-3}).
\end{align*}

\noindent{\bf Case (B): bounding the second term in Eqn. (\ref{eqn:mainresultcombined}).} 
Again, the arguments here proceed in the same way as in Case (B) of the proof of Theorem \ref{thm:insertiononlybadedge}, except for the conditional expectation of $\aY$ --- now $B_r^V(u)$ and $B_r^V(v)$ no longer span cliques but each of them spans a \myER{} random graphs with insertion probability $1-p$. Thus, we have :

\begin{align}\label{eqn:expectationIcombined}
\myE[\aY \mid \aH] &~~=~~ \sum\limits_{|S|=\tilde{k}+2}\myE[\aY_S \mid \aH]\nonumber\\
&~~=~~\sum\limits_{\substack{x_1 + x_2 = \tilde{k} \\ 0 \leq x_1 \leq N_u, 0 \leq x_2 \leq N_v}} \binom{N_u-1}{x_1} \binom{N_v-1}{x_2}q^{(x_1+1)(x_2+1)-1}(1-p)^{\binom{x_1+1}{2} + \binom{x_2+1}{2}}\nonumber\\
&~~=~~  \left[\sum\limits_{\substack{x_1 + x_2 = \tilde{k} \\ 0 \leq x_1 \leq N_u, 0 \leq x_2 \leq N_v}} \binom{N_u-1}{x_1} \binom{N_v-1}{x_2}\left(\frac{q}{1-p}\right)^{(x_1+1)(x_2+1)-1}\right](1-p)^{\frac{\tilde{k}^2+3\tilde{k}}{2}}
\end{align}

\begin{lemma}\label{lem:combinedcaseB}
For any constant $\epsilon>0$, we have that $\myE[\aY \mid \aH] = O(n^{-\epsilon})$ as long as the following condition on $p,q$ holds: 
\begin{align}\label{eqn:combinedqbound2}
\frac{q}{1-p} ~~\leq~~ &\min \Bigg\{\left(\dfrac{\tilde{k}!}{{\tilde{k}}^2n^{\epsilon}(N_u+N_v)^{\tilde{k}}}\right)^{\frac{1}{\tilde{k}}}\left(\frac{1}{1-p}\right)^{\frac{\tilde{k}^2+3\tilde{k}}{2}\frac{1}{\tilde{k} + \frac{\tilde{k}^2}{4}}}~~,\nonumber\\ 
&\left(\frac{\tilde{k}!}{(N_u+N_v)^{\tilde{k}}n^{\epsilon}}\right)^{\frac{16}{\tilde{k}^2}}\left(\frac{1}{1-p}\right)^{\frac{\tilde{k}^2+3\tilde{k}}{2}\frac{1}{\tilde{k} + \frac{\tilde{k}^2}{16}}} \Bigg\}.
\end{align}

Specifically, setting $\epsilon=3$ (a case which we will use later), we have $\myE[\aY \mid \aH] = O(n^{-3})$.
\end{lemma}
Again, the proof of the above lemma follows almost exactly the argument for Lemma \ref{lem:insertioncaseB} where, roughly speaking, (1) we will now use $\frac{q}{1-p}$ for the role of ``$q$" as in (\ref{eqn:expectationI}); and (2) compared to Eqn. (\ref{eqn:qbound2}) in Lemma \ref{lem:insertioncaseB}, the extra terms related to $\left( \frac{1}{1-p}\right)$ in the right hand side of Eqn. (\ref{eqn:combinedqbound2}) is needed to compensate the extra term $(1-p)^{\frac{\tilde{k}^2+3\tilde{k}}{2}}$ in the right hand side of Eqn. (\ref{eqn:expectationIcombined}) (as compared to Eqn. (\ref{eqn:expectationI}). Other than these two points, the proof proceeds in the same as as that for Lemma \ref{lem:insertioncaseB}. 
We thus omit the tedious details. 

Furthermore, it is easy to verify that
\begin{align*}
\left(\frac{1}{1-p}\right)^{\frac{3}{2}} < \min \left\{\left(\frac{1}{1-p}\right)^{\frac{\tilde{k}^2+3\tilde{k}}{2}\frac{1}{\tilde{k} + \frac{\tilde{k}^2}{4}}}, \left(\frac{1}{1-p}\right)^{\frac{\tilde{k}^2+3\tilde{k}}{2}\frac{1}{\tilde{k}+\frac{\tilde{k}^2}{16}}}    \right\}
\end{align*}
Thus, following the same argument as in the proof for (case-B) of Theorem \ref{thm:insertiononlybadedge}, we can obtain that there exist constants $c_1^b, c_2^b, c_3^b$ such that 
\begin{align*}
\text{If } q \leq \min\left\{c_1^b, c_2^b\cdot \left(\frac{1}{n} \right)^{c_3^b/\aK}\cdot \frac{\aK}{\mys n\sqrt{1-p}} \right\}, ~\text{then}~\myprob\left[ \Ghat|_{B_{uv}}~\text {has a}~uv\text{-clique of size} \ge \frac{\aK}{2}\right] = O(n^{-3}).
\end{align*}

Putting these two cases together, we have that there exist constants $c_1=\min\{c_1^a, c_1^b\}, c_2 = \min\{c_2^a, c_2^b\}$ and $c_3 = \max\{c_3^a,c_3^b \}$ such that if $p, q$ satisfy (\ref{eqn:qboundcombined}), then 
\begin{align*}
\myprob\left[\Ghat ~\text{has a}~uv\text{-clique of size} \ge \aK\right] = O(n^{-3})
\end{align*}

Finally, by applying the union bound, we know:
\begin{align*}
\myprob\left[\text{for every \mybadE{} } (u,v) \text{, } \Ghat ~\text{has a}~uv\text{-clique of size} \ge \aK\right]  = O(n^{-1})
\end{align*}
\end{proof}

\section{Proof of Theorem \ref{thm:maincombined}}
\label{appendix:thm:maincombined}
Our goal is to show that $\frac{1}{3}d_\anotherGt \le d_\tG \le 3d_\anotherGt$. 
Let $\mathcal{E}_1$ denote the event where $d_{\Ghat\cap \tG} \le 2d_\tG$. By Lemma 17 of \cite{ParthasarathyST17}, event $\mathcal{E}_1$ happens with probability at least $1 - n^{-\Omega(1)}$. 

Let $\mathcal{E}_2$ denote the event where all edges $\Ghat \cap \tG$ are also contained in the edge set of the filtered graph $\anotherGt$; that is, $\Ghat \cap \tG \subseteq \anotherGt$. 
By Lemma \ref{lm:combinedkeepgoodedges}, event $\mathcal{E}_2$ happens with probability at least $1-n^{-\frac{2}{3}}$ (this bound is derived in the proof of Theorem \ref{thm:combinedgoodclique}).
It then follows that:   
\begin{align}
\text{If both events } \mathcal E_1 \text{~and~} \mathcal E_2 \text{~happen, then~} d_\anotherGt \le d_{\Ghat \cap \tG} \le 2 d_\tG \leq 3 d_\tG. \nonumber
\end{align}

What remains is to show $d_\tG \le 3 d_\anotherGt$. 
To this end, we define $\mathcal E_3$ to be the event where for all \mybadE{}s $(u,v)$ in $\Ghat$, we have $\omega_{u,v}(\Ghat) < \tau$. If $\mathcal E_3$ happens, then it implies that for an arbitrary edge $(u, v) \in E(\anotherGt)$, either $(u, v) \in E(\tG)$ (thus $d_\tG(u,v) = 1$) or $d_\tG(u, v) \le 3$ (since there is at least one edge connecting $N_\tG(u)$ and $N_\tG(v)$).
By Lemma \ref{lm:combineddeletebadedges}, event $\mathcal E_3$ happens with probability at least $1- o(1)$ (the exact bound can be found in the proof of Theorem \ref{thm:combinedbadclique}). 

By applying the union bound, we know that $\mathcal E_1$, $\mathcal E_2$ and $\mathcal E_3$ happen simultaneously with high probability. 

Using a similar argument as the proof of Theorem 11 in \cite{ParthasarathyST17}, it then follows that given any $u, v\in V$ connected in $\anotherGt$, we can find a path in $\tG$ of at most $3 d_\anotherGt(u, v)$ number of edges to connect $u$ and $v$. 
Furthermore, event $\mathcal E_1$ implies that if $u$ and $v$ are not connected in $\anotherGt$, then they cannot be connected in $\tG$ either. 
Putting everything together, we thus obtain $d_\tG \le 3 d_\anotherGt $. Theorem \ref{thm:maincombined} then follows. 

\end{document}